\documentclass[sigconf]{acmart}
\AtBeginDocument{%
  }

\usepackage{graphicx}
\usepackage{subcaption}
\usepackage{multirow}
\usepackage{hyperref}

\usepackage{fancyhdr}
\copyrightyear{2025}
\acmYear{2025}
\setcopyright{acmlicensed}\acmConference[CIKM '25]{Proceedings of the 34th ACM International Conference on Information and Knowledge Management}{November 10--14, 2025}{Seoul, Republic of Korea}
\acmBooktitle{Proceedings of the 34th ACM International Conference on Information and Knowledge Management (CIKM '25), November 10--14, 2025, Seoul, Republic of Korea}
\acmDOI{10.1145/3746252.3761008}
\acmISBN{979-8-4007-2040-6/2025/11}

\hypersetup{
    colorlinks=true, 
    urlcolor=magenta,   
}

\begin{document}

\title{On the Cross-type Homophily of Heterogeneous Graphs: Understanding and Unleashing}

\author{Zhen Tao}
\affiliation{
  \institution{Nanjing University}
  \city{Nanjing}
  \country{China}
}
\email{zhentao.tz@gmail.com}

\author{Ziyue Qiao}
\affiliation{%
  \institution{Great Bay University}
  \city{Dongguan}
  \country{China}}
\email{ziyuejoe@gmail.com}

\author{Chaoqi Chen}
\affiliation{%
  \institution{Shenzhen University}
  \city{Shenzhen}
  \country{China}
}
\email{cqchen1994@gmail.com}

\author{Zhengyi Yang}
\affiliation{%
 \institution{University of New South wales}
 \city{Sydney}
 \country{Australia}}
\email{zhengyi.yang@unsw.edu.au}

\author{Lun Du}
\affiliation{%
  \institution{Ant Research}
  \city{Beijing}
  \country{China}}
\email{dulun.dl@antgroup.com}

\author{Qingqiang Sun$^*$}
\affiliation{%
  \institution{Great Bay University}
  \city{Dongguan}
  \country{China}}
\email{qqsun@gbu.edu.cn}
\thanks{$^*$Corresponding author.}

\renewcommand{\shortauthors}{Zhen Tao et al.}


\begin{abstract}
Homophily, the tendency of similar nodes to connect, is a fundamental phenomenon in network science and a critical factor in the performance of graph neural networks (GNNs). While existing studies primarily explore homophily in homogeneous graphs, where nodes share the same type, real-world networks are often more accurately modeled as heterogeneous graphs (HGs) with diverse node types and intricate cross-type interactions. This structural diversity complicates the analysis of homophily, as traditional homophily metrics fail to account for distinct label spaces across node types. To address this limitation, we introduce the Cross-Type Homophily Ratio (CHR), a novel metric that quantifies homophily based on the similarity of target information across different node types. Additionally, we propose \textit{Cross-Type Homophily-guided Graph Editing (CTHGE)}, a novel method for improving heterogeneous graph neural networks (HGNNs) performance by optimizing cross-type connectivity using Cross-Type Homophily Ratio.
Extensive experiments on five HG datasets with nine HGNNs validate the effectiveness of CTHGE, which delivers a maximum relative performance improvement of over 25\% for HGNNs on node classification tasks, offering a fresh perspective on cross-type homophily in HGs learning.
\vspace{-5pt}  
\end{abstract}

\ccsdesc[500]{Computing methodologies~Neural networks}
\vspace{-18pt}  
\keywords{Heterogeneous Graph Neural Networks, Homophily Ratio, Graph Editing}
\vspace{-18pt}  



\vspace{-18pt} 

\maketitle
\section{Introduction}
Homophily—the tendency for similar nodes to connect—is a well-established phenomenon in network science~\cite{pei2020geom, sun2024interdependence, luan2020complete, luan2024heterophilic}. Traditionally, it is defined on homogeneous graphs and quantified based on the similarity of node labels~\cite{sun2025single, pei2020geom,zheng2022graph}. Numerous studies have shown that Graph Neural Networks (GNNs)\cite{pei2020geom,zhang2024tatkc} perform better on graphs with high homophily \cite{lozares2014homophily, zhu2020beyond, luan2024heterophilic}.
However, real-world networks are often better modeled as heterogeneous graphs (HGs) which contain multiple types of nodes and edges to represent diverse and complex relationships~\cite{wang2022survey}. In HGs, the variation in node types and label spaces poses challenges for directly applying traditional homophily definitions (Fig.~\ref{fig:intro}). Moreover, cross-type edges---prevalent in many HGs---play a crucial role in structuring the network, yet existing definitions often overlook them by focusing solely on label similarity among same-type node connections. This reveals a significant gap between current homophily research in homogeneous settings and its application to heterogeneous graphs. 

\begin{figure}[t]
    \vspace{8pt}  
	\centering
	\includegraphics[width=0.99\linewidth]{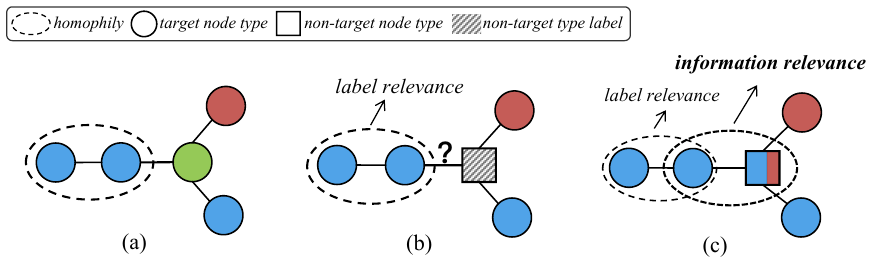}
    \vspace{-5pt}  
    \caption{Graph homophily in diverse cases where colors denote node labels: (a) the homophily in homogeneous graphs is considered based on node labels; (b) in HGs, label-relevance-based methods can only assess same-type homophily while failing to measure cross-type edge; (c) our method is capable of evaluating cross-type homophily in HGs by leveraging information relevance.}
	\label{fig:intro}
    \vspace{-18pt}  
\end{figure}
Could there exist a heterogeneity-aware metric for assessing Heterogeneous Graph Neural Network (HGNN) performance? Recent studies have shown that simplifying the structure of HGs by analyzing only same-type edges can influence HGNN effectiveness~\cite{luan2024heterophilic, guo2023homophily, li2023hetero}. However, such approaches neglect the substantial presence and potential informativeness of cross-type edges. Studying homophily from a cross-type perspective can provide complementary insights to improve label prediction, yet existing methods lack the capacity to effectively capture this aspect.

This paper aims to fill that gap by focusing on cross-type edges and introducing a novel formulation of cross-type homophily. To the best of our knowledge, this is the first work to systematically explore homophily in HGs from this perspective and to propose an effective solution.
To formalize the problem, we identify two core research questions that lie at the heart of cross-type homophily in HGs: 1) \emph{How can cross-type homophily be formally defined and meaningfully interpreted in heterogeneous graphs?} 2) \emph{How can this cross-type homophily be effectively and efficiently leveraged to improve the performance of heterogeneous graph neural networks (HGNNs)?}

To address the first challenge, we define Cross-Type Homophily Ratio (CHR), which distinguishes from traditional homophily metrics and quantifies homophily ratio between cross-type edges in HGs for the first time.  As shown in Fig.~\ref{fig:intro}(c), cross-type homophily broadens the scope of homophily research while also providing new perspectives for understanding the complex structures of HGs. Building on the principle of similarity-driven connections~\cite{mcpherson2001birds}, we proposed the concept of cross-type homophily and leverage target labels as a measure of target information to compute cross-type homophily ratio(CHR). Empirically, we constructed HGs with varying CHR values, confirming a positive correlation between CHR and HGNN performance. Theoretically, we analyzed and established the relationship between CHR and the generalization of HGNNs, demonstrating that an increase in CHR improves the generalization lower bound of HGNNs.

Regarding the second challenge, we propose the \underline{C}ross-\underline{T}ype \underline{H}omophily-guided \underline{G}raph \underline{E}diting (\textbf{CTHGE}), a novel method for improving HGNNs performance by optimizing cross-type connectivity using CHR. CTHGE operates in two phases: Phase I focuses on \textit{CHR-increasing Semantic Pruning} mechanism, where semantically misaligned edges are selectively removed, aiming to enhance CHR. Phase II consists of two components. The first integrates a \textit{Target-Driven Auxiliary Learning Paradigm} with confidence-based edge ranking. This approach leverages pre-training and fine-tuning to propagate label information from target-type nodes to their non-target neighbors, enabling pseudo-label supervision and transferring semantic knowledge across node types. The second component, \textit{Iterative Logits Refinement}, recalibrates semantic scores and pseudo-labels across multiple rounds, progressively refining edge selection and improving semantic consistency to align more accurately with target labels. The complexity analysis highlights the computational efficiency of our methods. As a plug-and-play framework, CTHGE can be integrated with various HGNNs with little effort. We implement CTHGE on top of nine representative HGNNs and evaluate their performance across five HG datasets, showcasing the effectiveness and adaptability of CTHGE. 

In a nutshell, our contributions can be highlighted as follows:

\begin{itemize}
\item[$\bullet$] We formally introduce Cross-Type Homophily as a novel concept for HGs and propose CHR as a novel metric to measure it, marking the first study to investigate homophily through cross-type edges in heterogeneous graphs.
\item[$\bullet$] We conduct theoretical analyses and empirical evaluations to investigate the connection between the CHR of HGs and classification performance of HGNNs, providing a comprehensive understanding of CHR.
\item[$\bullet$] We introduce CTHGE, an efficient graph editing strategy that leverages CHR to refine HG structure and significantly boosts HGNNs performance as a plug-and-play framework.
\item[$\bullet$] We conduct extensive experiments on five HG datasets with nine HGNNs. The results demonstrate the effectiveness of CTHGE which boosts the performance of diverse HGNNs by margins of up to 25\% across all datasets.

\end{itemize}

\section{Related Work}
\noindent\textbf{Heterogeneous Graph Neural Networks.} Recent advancements in HGNNs~\cite{zhou2020graph} can be categorized into meta-path-based and meta-path-free methods. Meta-path-based HGNNs rely on predefined meta-paths to define relationships. HAN~\cite{wang2019heterogeneous} transforms an HG into homogeneous subgraphs and uses attention for integration. GTN~\cite{yun2019graph} automates meta-path generation with learnable weights, and MEGNN~\cite{chang2022megnn} optimizes GTN’s efficiency. Meta-path-free methods focus on message passing and aggregation. RGCN~\cite{schlichtkrull2018modeling} aggregates neighbors by edge types, HGT~\cite{hu2020heterogeneous} uses a Transformer-based encoder. Simple-HGN~\cite{lv2021we} extends GAT with edge-type attention. HINormer~\cite{mao2023hinormer} combines local and heterogeneous relation encoders with GATv2~\cite{brody2021attentive} for improved node representation. These studies overlook the identification of HG characteristics that enhance model performance. Our work aims to explore these properties to further improve existing HGNNs.

\noindent\textbf{Heterophily in Graph Neural Networks.} Heterophily, the absence of homophily~\cite{lozares2014homophily}, is a significant challenge for GNNs~\cite{bo2021beyond, luan2022revisiting, chien2020adaptive, chen2023dirichlet,yan2022two}. Recent work has introduced several models and analyses to address this, such as Geom-GCN~\cite{pei2020geom} which uses network embeddings for structured neighborhoods, GBK-GNN~\cite{du2022gbk} which adjusts edge weights based on correlation blocks. Other approaches include UGCN~\cite{jin2021universal} with ranking-based aggregation and GloGNN~\cite{li2022finding} which integrates global node information. While evaluating heterophily in HGs is more complex due to the diversity of node types and interactions~\cite{luan2024heterophilic, ahn2022descent, xiong2023ground}. HDHGR~\cite{guo2023homophily} and Hetero2Net~\cite{li2023hetero} extend meta-path-based metrics for HGs. However, these approaches primarily adapt homogeneous graph principles, and overlook the issue of homophily across cross-type nodes in HGs. We propose a method that effectively handles homophily across cross-type nodes, providing a novel approach tailored for HGs.

\noindent\textbf{Graph Structure Learning and Graph Rewiring.} Graph structure learning (GSL) methods~\cite{jin2020graph,chen2020iterative,sun2022graph, yu2021graph, zheng2020robust,yang2019topology} refine graph structures from noisy data to uncover underlying relationships. These methods adjust node connections using metric learning~\cite{wang2020gcn}, probabilistic modeling~\cite{sun2022graph}, and direct optimization~\cite{jin2020graph}. GSL applications have also been extended to HGs with methods like HGSL~\cite{zhao2021heterogeneous} and SUBLIME~\cite{liu2022towards}, which uses contrastive learning. Graph rewiring improves connectivity for better training and inference efficiency~\cite{bi2024make,kenlay2021stability,guo2023homophily}, and graph pruning(GP) enhances GNN performance by removing irrelevant edges, preserving essential connectivity and reducing noise~\cite{rong2019dropedge, zheng2020robust, li2023less}. DropEdge~\cite{rong2019dropedge} and NeuralSparse~\cite{zheng2020robust} are popular GP methods. STEP~\cite{li2023less} offers a self-supervised pruning method for dynamic graphs. Our method performs graph editing by evaluating edge similarity to select edges, thereby enhancing the cross-type homophily of HGs and improving HGNN node classification performance.
\section{Preliminary Concepts}

\begin{definition}[Heterogeneous Graph] \label{def-1} A heterogeneous graph (HG) is defined as $\mathcal{G} = (\mathcal{V}, \mathcal{E}, \phi, \psi)$, where $\mathcal{V}$ and $\mathcal{E}$ denote the sets of nodes and edges, respectively. Each node $v \in \mathcal{V}$ is assigned a type $\phi(v)$, and each edge $e \in \mathcal{E}$ has a type $\psi(e)$. The sets of node and edge types are denoted as $T_v = {\phi(v) \mid v \in \mathcal{V}}$ and $T_e = {\psi(e) \mid e \in \mathcal{E}}$. When $|T_v| = |T_e| = 1$, $\mathcal{G}$ reduces to a homogeneous graph. \end{definition}

\begin{definition}[Homophily Ratio in Homogeneous Graphs] \label{def-2} Given a homogeneous graph $\mathcal{G} = (\mathcal{V}, \mathcal{E})$, the homophily ratio $H(\mathcal{G})$ quantifies the proportion of edges linking nodes with identical labels: \begin{equation}\label{eqn-HR} H(\mathcal{G}) = \frac{\sum_{(v_i, v_j) \in \mathcal{E}} \mathbb{I}(y_i = y_j)}{|\mathcal{E}|}, \end{equation} where $\mathbb{I}(\cdot)$ is the indicator function, and $y_i$, $y_j$ denote the labels of nodes $v_i$, $v_j$, respectively. \end{definition}

While homogeneous graphs involve a single edge type, heterogeneous graphs typically contain diverse cross-type edges (see Fig.~\ref{fig:edge_distribution_of_HGs}). However, homophily as defined in homogeneous settings does not trivially extend to HGs, posing challenges in defining and measuring cross-type homophily.

\begin{figure}[t]
	\centering
	\includegraphics[width=0.90\linewidth]{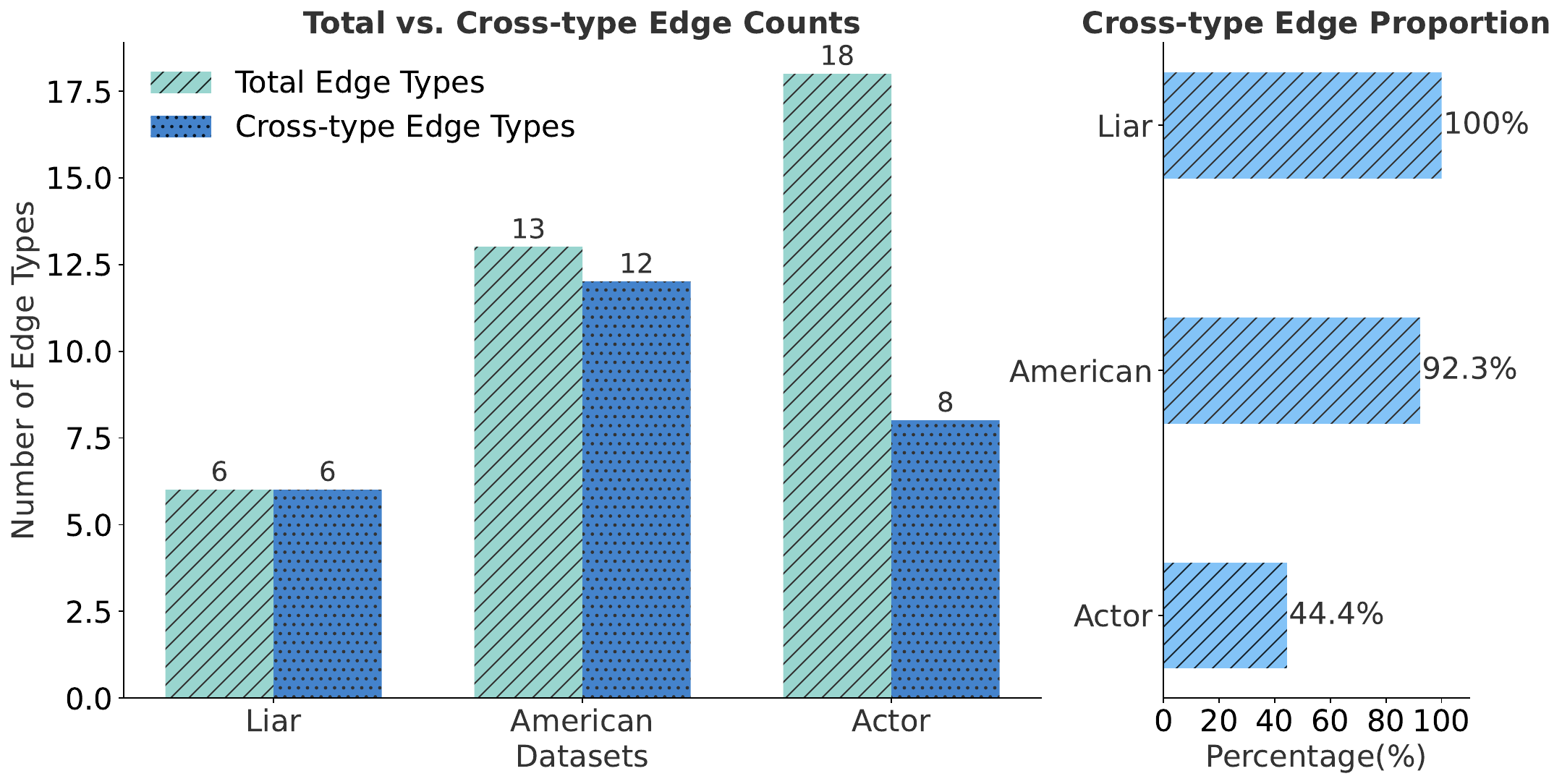}
	\caption{Edge type distribution across different heterogeneous graphs. The proportion of {Cross-type Edge Types} among the edge types in HGs is significant.}
	\label{fig:edge_distribution_of_HGs}
\end{figure}

\section{Understanding and Unleashing the Cross-type Homophily of HGs}
\subsection{Understanding of Cross-Type Homophily}
\label{subsection:Understanding HGs}
This subsection introduces the Cross-Type Homophily Ratio (CHR), a formal metric designed to quantify homophily across heterogeneous node types. We further provide both empirical and theoretical analyses to demonstrate its effectiveness in characterizing HGNN generalization and performance.

We begin by formalizing homophily in heterogeneous graphs through structural decomposition of node and edge types.

\begin{definition}[Node and Edge Types in Heterogeneous Graphs]
\label{def-3}
Let \( \mathcal{G} = (\mathcal{V}, \mathcal{E}, \phi, \psi) \) denote an undirected heterogeneous graph, where nodes are partitioned into target nodes \( \mathcal{V}_t \) and non-target nodes \( \mathcal{V}_n \), such that \( \mathcal{V} = \mathcal{V}_t \cup \mathcal{V}_n \). Edges are categorized into: (1) intra-target edges \( \mathcal{E}_{tt} \) between \( \mathcal{V}_t \), (2) cross-type edges \( \mathcal{E}_{tn} \) between \( \mathcal{V}_t \) and \( \mathcal{V}_n \), and (3) intra-non-target edges \( \mathcal{E}_{nn} \) between \( \mathcal{V}_n \). Corresponding adjacency matrices are denoted as \( \mathcal{A}_{tt} \), \( \mathcal{A}_{tn} \), and \( \mathcal{A}_{nn} \), with \( \mathcal{A}_{tn} = \mathcal{A}_{nt}^\top \) due to graph symmetry.
\end{definition}

\noindent\textbf{Target Information Representation.}  
Each node \( u \in \mathcal{V} \) is assigned a target information vector \( \mathcal{H}_u = [h_1, h_2, \dots, h_c] \), where \( c \) denotes the number of target categories, and \( h_i \) quantifies the degree of association with the \( i \)-th category.

\noindent\textbf{Target Information Initialization.}  
We construct the initial target information matrix \( L \in \mathbb{R}^{N_t \times C} \) for target nodes based on label availability. For each labeled node \( v_i^{\text{train}} \in \mathcal{V}_{t\text{-}train} \), the label is encoded as a one-hot vector:  
\(
L(v_i^{\text{train}}) = \mathbf{e}_{y_i}, y_i \in \{1, \dots, C\},
\)  
where \( \mathbf{e}_{y_i} \in \mathbb{R}^C \) denotes the standard basis vector for class \( y_i \).
For unlabeled target nodes \( v_i^{\text{test}} \in \mathcal{V}_{t\text{-}test} \), we adopt a soft-labeling strategy to reflect uncertainty in class membership. Specifically, we train a heterogeneous GNN \( f_\theta \) on \( \mathcal{G} \), and use its output logits to estimate a class probability distribution via softmax:
\begin{equation}\label{eqn-L_test} 
L(v_i^{\text{test}}) = \text{softmax}(f_\theta(\mathcal{G}, \mathcal{X})[v_i^{\text{test}}]) \in \mathbb{R}^{C}.
\end{equation}

This approach preserves uncertainty and provides a smoother initialization for downstream homophily modeling.
The matrix \( L \) serves as the initial target information for all target nodes.

\noindent\textbf{Target Information Propagation.}  
Target information is propagated to non-target nodes via cross-type edges \( \mathcal{E}_{tn} \) as:
\begin{equation}\label{eqn-P} 
  P = (\mathcal{W} \circ \mathcal{A}_{nt}) L,
\end{equation}
where \( P \in \mathbb{R}^{N_n \times C} \) is the propagated matrix for non-target nodes, \( \mathcal{W} \) is based on the weight of each edge in the dataset, and \( \circ \) denotes the Hadamard product.

To ensure comparability between \( L \) and \( P \), we normalize the information associated with each node as a probability distribution using the \( \ell_1 \)-norm:
\begin{equation}\label{eqn-P_i} 
  P_i^{\prime} = \frac{P_i}{\| P_i \|_1}.
\end{equation}

The normalized target information matrix for non-target nodes $ P^{\prime} \in \mathbb{R}^{N_n \times C} $ is concatenated with the target information matrix $ L \in \mathbb{R}^{N_t \times C} $ to form the final target information matrix:
\begin{equation}\label{eqn-H_cat} 
    \mathcal{H} = \text{concat}\left( L, P^{\prime} \right) \in \mathbb{R}^{N \times C},
\end{equation}
where $ N = N_t + N_n $ represents the total number of nodes. 

Based on the preceding formulation, we formally define the Cross-Type Homophily Ratio (CHR) in the context of HGs.

\begin{figure}[t]
  \centering
  
  \begin{subfigure}{0.48\linewidth}
		\centering
		\includegraphics[width=1\linewidth]{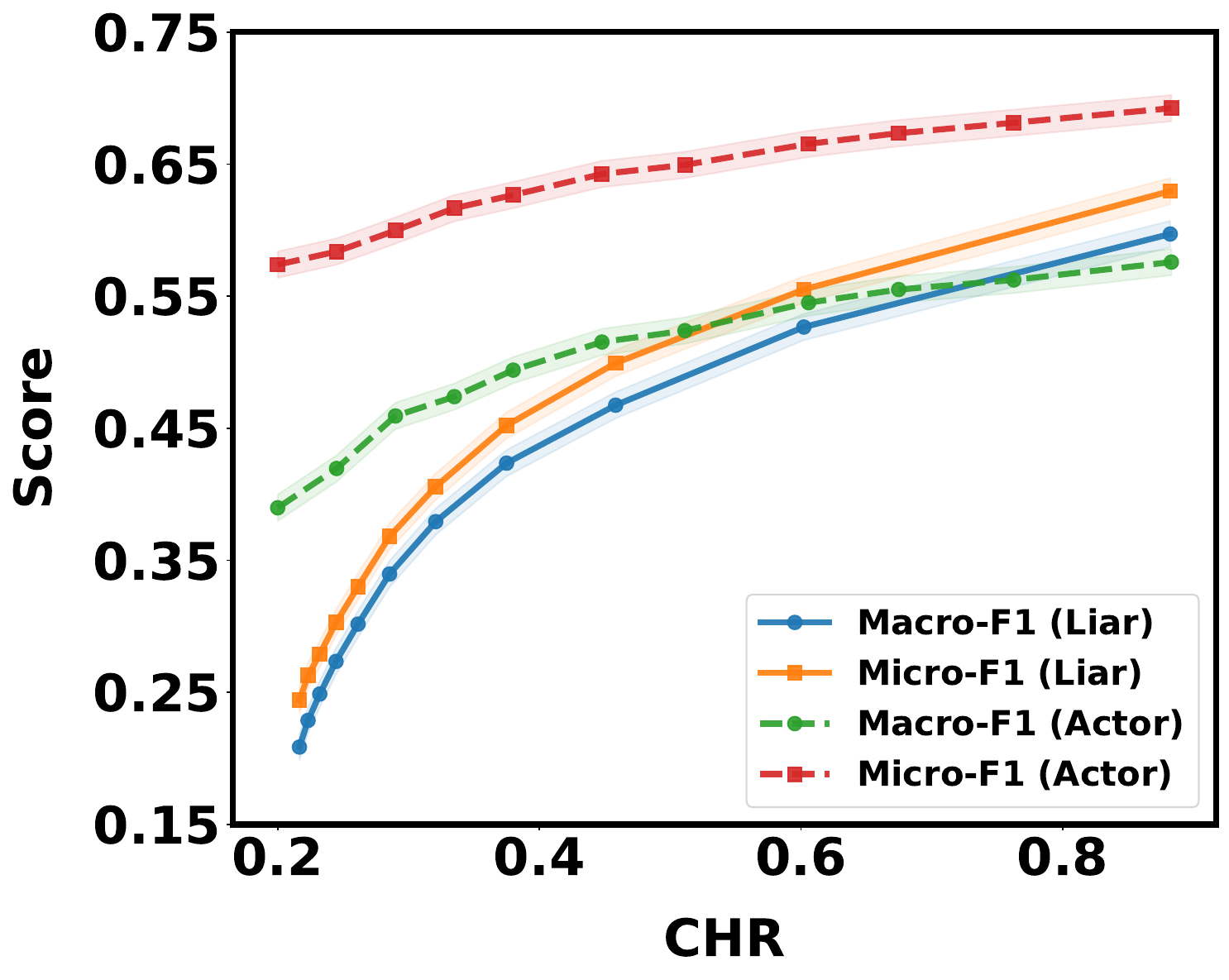}
		\caption{Liar \& Actor}
		\label{fig-ob_CHR_F1_Liar_Actor}
	\end{subfigure}
  \hfill
  \begin{subfigure}{0.48\linewidth}
		\centering
		\includegraphics[width=1\linewidth]{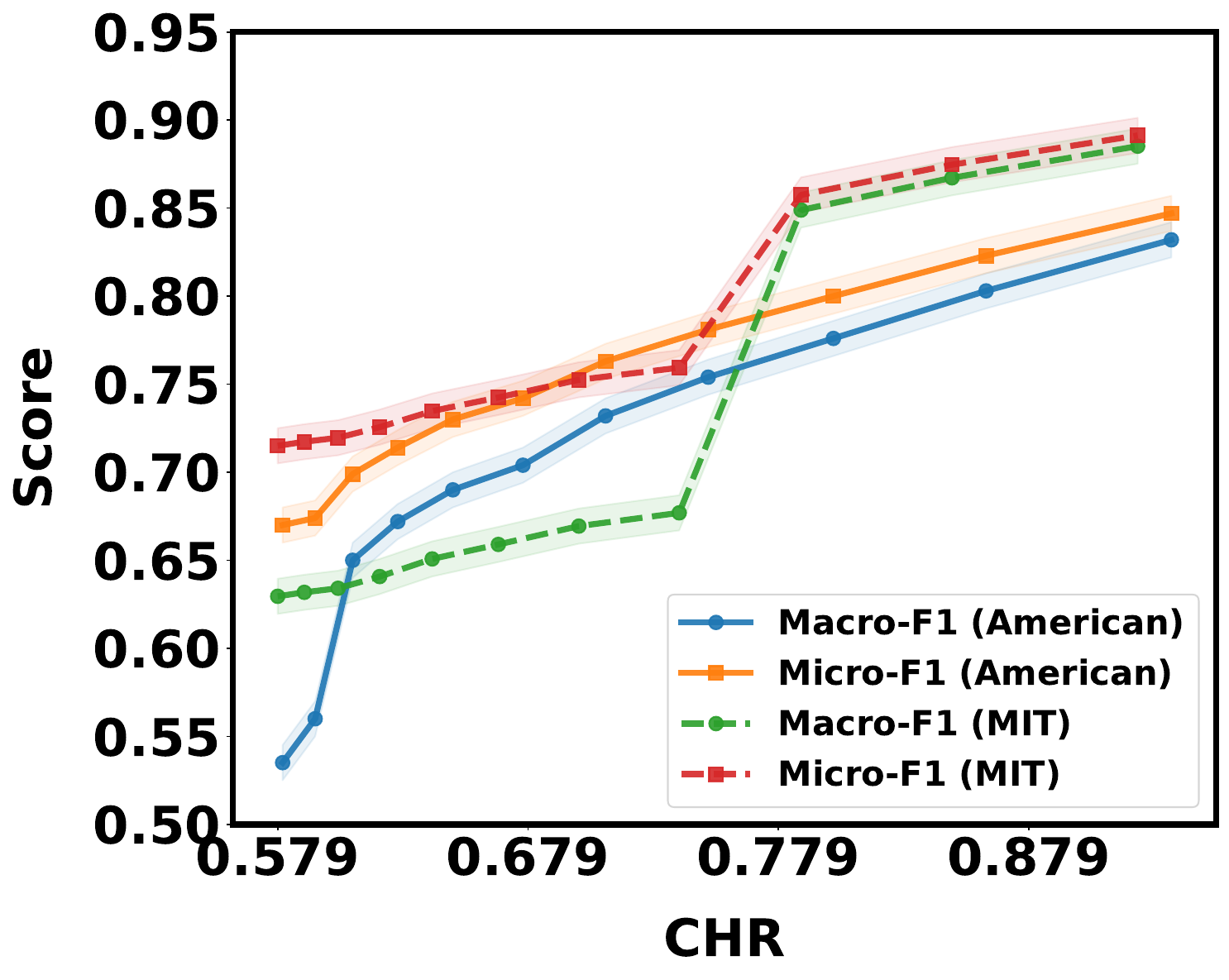}
		\caption{American \& MIT}
		\label{fig-ob_CHR_F1_American_MIT}
	\end{subfigure}
  \caption{Impact of CHR on HGNN Performance.}
  \label{fig:ob_CHR_F1_hr}
\end{figure}

\begin{definition}[Cross-Type Homophily Ratio]
	\label{def-4}
    The Cross-Type Homophily Ratio (CHR) quantifies the homophily between target nodes and non-target nodes in a heterogeneous graph \(\mathcal{G}\). In this work, we use the target information matrix \( \mathcal{H} \in \mathbb{R}^{N \times C} \) to measure this cross-type homophily indicator. The CHR in a heterogeneous graph is calculated as:
\begin{equation}\label{eqn-CHR} 
CHR(\mathcal{G}) = \frac{1}{|\mathcal{E}_{\text{tn}}|} \sum\nolimits_{(v_i, v_j) \in \mathcal{E}_{\text{tn}}} \text{sim}(v_i, v_j),
\end{equation}
where \( \text{sim}(v_i, v_j) \) is represented as:
\begin{equation}\label{eqn-sim} 
\text{sim}(v_i, v_j) = S_{ij} = \sum\nolimits_{k=1}^{C} \mathcal{H}(v_{i})_k \mathcal{H}(v_{j})_k,
\end{equation}
here, \( S_{ij} \) is the similarity score between the target information of two nodes based on their respective categories.
\end{definition}

\noindent\textbf{Empirical Study of CHR and HGNN Performance.}
To assess the relationship between CHR and model effectiveness, we conduct a series of experiments on HGs with varying CHR levels, generated by systematically modifying graph structure. Each HG is split into training, validation, and test sets under a fixed ratio, and a unified HGNN model is trained across all instances. As shown in Fig.\ref{fig:ob_CHR_F1_hr}, the results demonstrate a strong positive correlation between CHR and node classification performance. This finding indicates that enhancing homophily between target and non-target nodes facilitates the HGNN’s ability to capture label consistency, thereby substantially boosting classification effectiveness.

\noindent\textbf{Theoretical Analysis of CHR and HGNN Generalization.}  
We theoretically investigate the effect of cross-type homophily on the generalization ability of HGNNs by employing a complexity measure (CM) framework. Specifically, we adopt the representation consistency approach~\cite{natekar2020representation}, which utilizes the Davies-Bouldin index~\cite{davies1979cluster} to quantify model complexity.
Formally, the complexity measure is defined as a function \( M: \{H, S\} \to \mathbb{R}^+ \), where \( H \) denotes the class of meta-path-free HGNNs with different parameters, and \( S \) represents heterogeneous graphs with varying CHR levels. For a given layer and dataset, intra-class and inter-class distances are computed as:
\begin{equation}\label{eqn-si}
S_{i} = \left(\frac{1}{n_{k}} \sum\nolimits_{t=1}^{n_{i}} |O_{i}^{(t)} - \mu_{O_{i}}|^{p} \right)^{\frac{1}{p}}, \quad i = 1, \dots, k
\end{equation}
\begin{equation}\label{eqn-mij}
M_{i,j} = \|\mu_{O_{i}} - \mu_{O_{j}}\|_{p}, \quad i, j = 1, \dots, k
\end{equation}
where $ i, j $ are category indices, \( O_i^{(t)} \) is the representation of the \( t \)-th instance in class \( i \), and \( \mu_{O_i} \) is its centroid. The Davies-Bouldin index is then given by:
\begin{equation}
C = \frac{1}{k} \sum\nolimits_{i=1}^{k} \max_{j \ne i} \frac{S_i + S_j}{M_{i,j}}. \label{eqn-dbindex}
\end{equation}

Lower values of \( C \) indicate better separation and thus stronger generalization. Setting \( p=2 \), the metric reflects the ratio of intra-class variance to inter-class separation. Under this setting, we state the following theorem:

\begin{theorem}
Let $\mathcal{G} = (\mathcal{V}, \mathcal{E}, \phi, \psi)$ denote an HG. We consider a binary classification problem involving node classification with an HGNN across the entire graph $\mathcal{G}$. Using target information as the classification criterion, we model the distribution of non-target nodes as a spatial interpolation of the target node distributions. When Cross-Type Homophily reaches a maximum value of 1, the generalization capacity of the HGNN achieves its theoretical upper bound.
\end{theorem}

\begin{proof}
To establish this theorem, we introduce the parameter $ Q_c $ to estimate the lower bound of the complexity measure $ C $ within a general heterogeneous graph convolutional layer. Employing principles from Consistency of Representations~\cite{natekar2020representation} and Fisher Discriminant Analysis~\cite{fisher1936use}, we utilize the ratio of intra-class variance to inter-class variance as a primary metric, reformulating the Consistency of Representations as a squared expression to enable rigorous theoretical validation:
\begin{equation}\label{eqn-consistency}
C = \frac{1}{k} \sum_{t=0}^{k-1} \max_{t \neq s} \frac{T_t^2 + T_s^2}{M_{t,s}^2}.
\end{equation}

Our analysis focuses on node classification across HG by leveraging target information in a binary classification setting. For non-target nodes $ X_{nt} $, their representations are expressed as a spatial mixture of target node distributions $ X_0 $ and $ X_1 $:
\begin{equation}\label{eqn-combination}
\mu X_{nt} = \lambda \mu X_0 + (1-\lambda) \mu X_1.
\end{equation}

To account for homophily in nodes with mixed connections, we define Same-Type Homophily $ Q_s $ and Cross-Type Homophily $ Q_c $. The representations of target nodes $ \mu O_0 $ and $ \mu O_1 $ are derive as follows:

\begin{equation}\label{eqn-mu0}
\begin{aligned}
\mu_{O_0} &= \mathbb{E}\left(\mathbf{W} \sum_{j \in \mathcal{N}_{r}(v_{i})} \frac{1}{|\mathcal{N}_{r}(v_{i})|} \mathbf{X}^{(j)}\right) \\
&= \mathbf{W} \left(Q_s \mu_{\mathbf{X}_{0}} + (1-Q_s) \mu_{\mathbf{X}_{1}} + Q_c \mu_{\mathbf{X}_{0}} + (1-Q_c) \mu_{\mathbf{X}_{1}}\right) \\
&= \mathbf{W} \left((Q_s + Q_c) \mu_{\mathbf{X}_{0}} + (2 - Q_s - Q_c) \mu_{\mathbf{X}_{1}}\right),
\end{aligned}
\end{equation}

\begin{equation}\label{eqn-mu1}
\begin{aligned}
\mu_{O_1} &= \mathbb{E}\left(\mathbf{W} \sum_{j \in \mathcal{N}_{r}(v_{i})} \frac{1}{|\mathcal{N}_{r}(v_{i})|} \mathbf{X}^{(j)}\right) \\
&= \mathbf{W} \left(Q_s \mu_{\mathbf{X}_{1}} + (1-Q_s) \mu_{\mathbf{X}_{0}} + Q_c \mu_{\mathbf{X}_{1}} + (1-Q_c) \mu_{\mathbf{X}_{0}}\right) \\
&= \mathbf{W} \left((Q_s + Q_c) \mu_{\mathbf{X}_{1}} + (2 - Q_s - Q_c) \mu_{\mathbf{X}_{0}}\right).
\end{aligned}
\end{equation}

The inter-class distance $ M_{0,1} $ can be computed as:

\begin{equation}\label{eqn-m01}
\begin{aligned}
M_{0,1}^{2} &= \|\mu_{O_{0}} - \mu_{O_{1}}\|^{2} \\
&= \left\|\mathbf{W} \left((2Q_s + 2Q_c - 2) \mu_{\mathbf{X}_{0}} \right) + \mathbf{W} \left((2 - 2Q_s - 2Q_c) \mu_{\mathbf{X}_{1}} \right)\right\|^{2} \\
&= \left\|\left(2Q_s + 2Q_c - 2\right) \mathbf{W} \left(\mu_{\mathbf{X}_{0}} - \mu_{\mathbf{X}_{1}}\right)\right\|^{2}.
\end{aligned}
\end{equation}

Applying Jensen's Inequality yields:

\begin{equation}\label{eqn-jensen}
M_{0,1}^{2} \leq \left(2Q_s + 2Q_c - 2\right)^{2} \cdot \left\|\mathbf{W} \left(\mu_{\mathbf{X}_{0}} - \mu_{\mathbf{X}_{1}}\right)\right\|^{2}.
\end{equation}

The intra-class variances for classes 0 and 1 are computed as:

\begin{equation}\label{eqn-s0}
\begin{aligned}
T_{0}^{2} &= \mathbb{E}\left(\langle O_{0}^{(j)} - \mu_{O_{0}}, O_{0}^{(j)} - \mu_{O_{0}} \rangle \right) \\
&= \mathbb{E}\left(\left(Q_s + Q_c \right)^{2} \cdot (X_{0} - \mu_{X_{0}})^{T} \cdot \mathbf{W}^{T} \mathbf{W} \cdot (X_{0} - \mu_{X_{0}})\right) \\
&+ \mathbb{E}\left(\left(2 - Q_s - Q_c\right)^{2} \cdot (X_{1} - \mu_{X_{1}})^{T} \cdot \mathbf{W}^{T} \mathbf{W} \cdot (X_{1} - \mu_{X_{1}})\right).
\end{aligned}
\end{equation}

We introduce the substitutions $ F = (Q_s + Q_c)\mathbf{W} $, $ G = (2 - Q_s - Q_c)\mathbf{W} $, $ \Delta X_0 = X_0 - \mu X_0 $, and $ \Delta X_1 = X_1 - \mu X_1 $, allowing us to rewrite $ T_0^2 $ as:

\begin{equation}\label{eqn-s0-rewrite}
T_0^2 = \mathbb{E}\left((\Delta X_0)^T \cdot F^T F \cdot \Delta X_0 \right) + \mathbb{E}\left((\Delta X_1)^T \cdot G^T G \cdot \Delta X_1 \right).
\end{equation}

Similarly, $ T_1^2 $ becomes:

\begin{equation}\label{eqn-s1}
T_1^2 = \mathbb{E}\left((\Delta X_0)^T \cdot G^T G \cdot \Delta X_0 \right) + \mathbb{E}\left((\Delta X_1)^T \cdot F^T F \cdot \Delta X_1 \right).
\end{equation}

Utilizing the inequality $ x^T \cdot (F^T F + G^T G) \cdot x \geq x^T \cdot \left(\frac{1}{2} \cdot (F + G)^T (F + G)\right) \cdot x $, where equality holds if $ F = G $ and $ F + G = 2\mathbf{W} $, we obtain:

\begin{equation}\label{eqn-sum-s0-s1}
\begin{aligned}
T_0^2 + T_1^2 &\geq \frac{1}{2} \mathbb{E}\left[\Delta X_0^T \cdot \left(2\mathbf{W}\right)^T \left(2\mathbf{W}\right) \cdot \Delta X_0\right] \\
&+ \frac{1}{2} \mathbb{E}\left[\Delta X_1^T \cdot \left(2\mathbf{W}\right)^T \left(2\mathbf{W}\right) \cdot \Delta X_1\right].
\end{aligned}
\end{equation}

Therefore, the lower bound $ C_{\text{lower}} $ of the complexity measure $ C $ can be computed as:

\begin{equation}\label{eqn-clower-apx}
\begin{aligned}
C &= \frac{T_{0}^{2} + T_{1}^{2}}{M_{0,1}^{2}}  \\
  &\geq 2\cdot \frac{ \mathbb{E}\left[\Delta X_{0}^{T} \cdot \mathbf{W}\mathbf{W}^{T} \cdot \Delta X_{0}\right] +  \mathbb{E}\left[\Delta X_{1}^{T} \cdot \mathbf{W}\mathbf{W}^{T} \cdot \Delta X_{1}\right]}{\left(2Q_s + 2Q_c - 2\right)^{2} \cdot \left\|\mathbf{W}\left(\mu_{\mathbf{X}_{0}} - \mu_{\mathbf{X}_{1}}\right)\right\|^{2}}.
\end{aligned}
\end{equation}

The components independent of $ Q_c $ in both the numerator and denominator form a positive constant $ C_0 $:

\begin{equation}\label{eqn-c0}
C_{0} = 2\cdot \frac{ \mathbb{E}\left[\Delta X_{0}^{T} \cdot \mathbf{W}\mathbf{W}^{T} \cdot \Delta X_{0}\right] +  \mathbb{E}\left[\Delta X_{1}^{T} \cdot \mathbf{W}\mathbf{W}^{T} \cdot \Delta X_{1}\right]}{\|\mathbf{W}(\mu_{{\mathbf{X}_{0}}} - \mu_{{\mathbf{X}_{1}}})\|^{2}}.
\end{equation}

Thus, the lower bound $ C_{\text{lower}} $ becomes:

\begin{equation}\label{eqn-clower-final}
C^{\mathrm{lower}} = \frac{C_{0}}{(2Q_{s} + 2Q_{c} - 2)^{2}}.
\end{equation}

The partial derivative of $ C_{\text{lower}} $ with respect to $ Q_c $ is:

\begin{equation}\label{eqn-partial-derivative}
\frac{\partial C^{\mathrm{lower}}}{\partial Q_c} = -\frac{2C_0}{(2Q_s + 2Q_c - 2)^3}.
\end{equation}

As $ Q_c \to 1 $, the lower bound $ C_{\text{lower}} $ attains its minimum, signifying an optimal state of generalization. This result highlights that increasing cross-type homophily within a heterogeneous graph significantly enhances the generalization capability of HGNNs.

\end{proof}

\subsection{Unleashing CHR to Boost HGNNs}
In heterogeneous graph learning, the quality of cross-type connectivity plays a pivotal role in the effective transmission of task-relevant information. While node relations establish the graph’s topological structure, these connections often inadequately reflect the underlying label semantic consistency across heterogeneous node types, resulting in disrupted information flow and constrained model generalization. This phenomenon is quantitatively captured by the Cross-Type Homophily Ratio introduced earlier, which serves as a key metric to evaluate and guide the alignment between structural and label.

We propose a novel graph rewiring approach, named \textit{Cross-Type Homophily Guided Graph Editing (CTHGE)}. Unlike conventional methods that indiscriminately optimize the entire graph topology, CTHGE explicitly adopts CHR as an objective metric. By maximizing CHR, it selectively refines cross-type connections to enhance the propagation of label-consistent information, thereby improving both the representational capacity and generalization performance of HGNNs. Our method comprises two phase:

\begin{figure*}[t]
	\centering
	\includegraphics[width=0.97\linewidth]{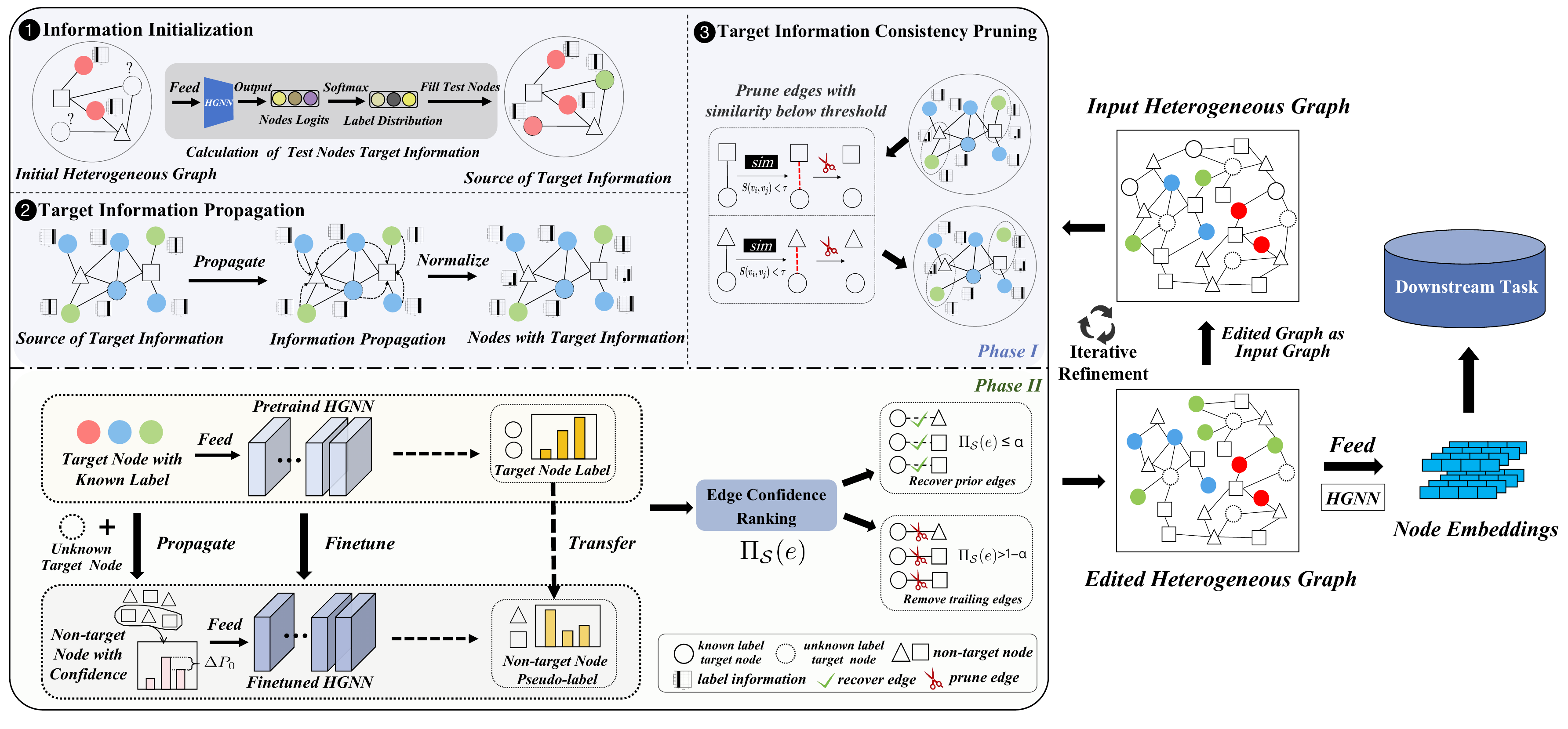}
	\caption{Overview of the CTHGE Method.}
	\label{fig:CTHGE}
\end{figure*}
\noindent\textbf{Phase I: CHR-Increasing Semantic Pruning.} We introduce a task-driven pruning mechanism that selectively removes low-consistency cross-type associations based on their semantic alignment, aiming to improve the overall Cross-Type Homophily Ratio(CHR) of the connectivity set.

Based on Equations~\ref{eqn-L_test} to \ref{eqn-H_cat}, we derive the target-informed distribution matrix $\mathcal{H} \in \mathbb{R}^{|\mathcal{V}| \times C}$ that encodes the label probability vector for each node. The similarity of any cross-type pair $(v_i, v_j) \in \mathcal{E}_{tn}$ is computed as:
\begin{equation}
\mathcal{S}(v_i, v_j) = \text{sim}(v_i, v_j) = \sum_{k=1}^C \mathcal{H}(v_i)_k \cdot \mathcal{H}(v_j)_k.
\end{equation}

To enforce minimal semantic consistency, we introduce a pruning decision function controlled by a tunable threshold $\tau \in [0,1]$:
\begin{equation}
\mathcal{F}_\tau(v_i, v_j) = \mathbb{I}[\mathcal{S}(v_i, v_j) \geq \tau].
\end{equation}

The retained association set is then:
\begin{equation}
\mathcal{E}_{\text{prune}} = \{(v_i, v_j) \in \mathcal{E}_{tn} \mid \mathcal{F}_\tau(v_i, v_j) = 1 \}.
\end{equation}

This pruning operation increases the average label-space similarity across surviving links, thus yielding a higher CHR value:
\begin{equation}
CHR(\mathcal{E}_{\text{prune}}) > CHR(\mathcal{E}_{tn}).
\end{equation}

\paragraph{Optimization Perspective.}
The above rule can be interpreted as the constrained maximization of CHR consistency over the link set. Introducing binary edge selection variables $\mathbf{z} = \{z_{(i,j)}\} \in \{0,1\}^{|\mathcal{E}_{tn}|}$, we obtain:
\begin{equation}
\max_{\mathbf{z}} \frac{1}{\sum z_{(i,j)}} \sum z_{(i,j)} \cdot \mathcal{S}(v_i, v_j), \quad \text{s.t. } z_{(i,j)} = \mathbb{I}[\mathcal{S}(v_i, v_j) \geq \tau].
\end{equation}

This formulation underscores that the pruning step performs a deterministic, non-parametric Cross-Type Homophily Ratio maximization under a semantic similarity constraint. The resulting subgraph is thus not topologically optimal perse, but optimally task-aligned in the label consistency space.

\noindent\textbf{Phase II: Confidence-Aware Refinement of Cross-Type Semantic Connectivity.}  
While the pruning phase improves Cross-Type Homophily Ratio by removing semantically misaligned links, it may also discard potentially connections valuable for task-specific signal propagation. To address this, we propose a refinement mechanism to increase the confidence of edge selection.

\textit{Target-Driven Auxiliary Learning Paradigm Task.} We reformulate cross-type consistency learning by propagating label signals from target-type nodes to their non-target neighbors, enabling pseudo-label supervision and transferring semantic knowledge across types. This paradigm captures homophily not only via observed labels, but also through model-inferred alignment with target semantics.

Formally, let \( L_0 \in \mathbb{R}^{|\mathcal{V}_t| \times C} \) denote the one-hot label matrix of target-type nodes. We define a linear propagation operator over cross-type adjacency \( \mathcal{A}_{nt} \), yielding the soft target distribution matrix for non-target nodes as:
\(
P_0 = \mathcal{A}_{nt} \cdot L_0 \in \mathbb{R}^{|\mathcal{V}_{nt}| \times C}.
\)

To quantify the reliability of each induced distribution, we define a confidence metric based on the dominance of the top class probability:
\begin{equation}
\Delta P_0(v_i) = \max P_0(v_i) - \text{second\_max}(P_0(v_i)).
\end{equation}

This value reflects the entropy margin and serves as a proxy for pseudo-labeling certainty. We designate a non-target node \( v_i \in \mathcal{V}_{nt} \) as confident if:
\begin{equation}
\Delta P_0(v_i) > \left| \{ v_j^{\text{test}} \mid (v_i, v_j^{\text{test}}) \in \mathcal{E}_{tn} \} \right|,
\end{equation}
and assign its class label as:
\begin{equation}
P_{\text{cl}}(v_i) = \arg\max P_0(v_i).
\end{equation}

We construct a two-stage transfer learning procedure, initializing the model \( f_\theta^{\text{pre}} \) on the primary supervised objective over labeled target nodes:
\begin{equation}
\mathcal{L}_{\text{pre}}(\theta) = - \sum_{v_i \in \mathcal{V}_t^{\text{train}}} L(v_i) \cdot \log f_\theta^{\text{pre}}(v_i).
\end{equation}

We then fine-tune the model on the confident subset of non-target nodes using their pseudo-labels:
\begin{equation}
\mathcal{L}_{\text{fine}}(\theta) = - \sum_{v_i \in \mathcal{V}_{nt}^{ct}} P_{\text{cl}}(v_i) \cdot \log f_\theta^{\text{fine}}(v_i).
\end{equation}

This paradigm enables knowledge transfer across semantic types by propagating label supervision, thereby enhancing the model’s inductive bias toward Cross-Type Homophily Ratio-consistent message passing.

Leveraging pseudo-labels from this process, we perform semantic similarity reconfiguration to refine cross-type edge selection. Specifically, we adopt a confidence-ranked edge evaluation scheme, in which the soft similarity between a non-target node \(v_i\) and a target node \(v_j\) is computed using class-wise inner products of calibrated logits and ground-truth labels:

\begin{equation}
\mathcal{S}^{\text{soft}}(v_i, v_j) = \sum_{k=1}^C \text{softmax}(f_\theta^{\text{fine}}(v_i))_k \cdot L(v_j)_k.
\end{equation}

We rank all cross-type edges \( e = (v_i, v_j) \in \mathcal{E}_{tn} \) based on their semantic alignment percentile:
\begin{equation}
\Pi_{\mathcal{S}}(e) = \frac{1}{|\mathcal{E}_{tn}|} \cdot \left| \left\{ e' \in \mathcal{E}_{tn} \mid \mathcal{S}^{\text{soft}}(e') > \mathcal{S}^{\text{soft}}(e) \right\} \right|.
\end{equation}

Given a threshold ratio \( \alpha\), we identify candidate edges for recovery and removal as:
\begin{align}
\mathcal{E}_{\text{rec}} &= \{ e \in \mathcal{E}_{\text{cand}} \mid \Pi_{\mathcal{S}}(e) \leq \alpha \} \\
\mathcal{E}_{\text{rem}} &= \{ e \in \mathcal{E}_{\text{prune}} \mid \Pi_{\mathcal{S}}(e) > 1 - \alpha \}.
\end{align}

The final edge set after semantic refinement is computed as:
\begin{equation}
\mathcal{E}_{\text{final}} = (\mathcal{E}_{\text{prune}} \setminus \mathcal{E}_{\text{rem}}) \cup \mathcal{E}_{\text{rec}},  \quad \mathcal{G} = (\mathcal{V}, \mathcal{E}_{\text{final}}).
\end{equation}

This procedure approximates the following structured subset selection problem:
\begin{equation}
\max_{\mathbf{z}} \sum_{e} z_e \cdot \mathcal{S}^{\text{soft}}(e) \quad \text{s.t. } \sum_{e} z_e = (1-\alpha) |\mathcal{E}_{\text{prune}}| + \alpha |\mathcal{E}_{\text{cand}}|.
\end{equation}

Our ranking-based selection mechanism ensures a balanced structural update aligned with semantic consistency as reflected in model predictions.

\textit{Iterative Logits Refinement.} To further stabilize semantic consistency and guard against local estimation noise, we incorporate a multi-round edge reconfiguration mechanism. At each iteration, semantic scores and pseudo-labels are recalibrated based on updated model predictions. A gap factor \( \gamma \) and base offset \( \mathcal{O} \) determine the proportion of edges subject to refinement in each round. This iterative process amplifies the influence of high-confidence semantic cues and allows the system to progressively converge toward a target label aligned connectivity manifold.

\vspace{5pt}  
\subsection{Complexity Analysis}

The proposed CTHGE consists of two phases with distinct computational demands. Phase I involves target information initialization, propagation, and pruning, with complexity dominated by \(O(N_t C + N_t F^2 L + N_n N_t C + |\mathcal{E}|)\), where \(N_t\), \(N_n\), \(F\), \(L\), and \(|\mathcal{E}|\) denote the numbers of target nodes, non-target nodes, hidden units, network layers, and edges, respectively. Phase II centers on HGNN pretraining and fine-tuning over training nodes \(N\), semantic similarity scoring for cross-type edges, sorting, and iterative logits refinement. Its complexity can be approximated by \(O\bigl  (N F^2 L + |\mathcal{E}_{tn}| C + |\mathcal{E}_{tn}| \log |\mathcal{E}_{tn}|\bigr)\), where \(|\mathcal{E}_{tn}|\) the number of cross-type edges. The computational bottleneck lies primarily in HGNN training and information propagation steps, and scale with iteration count.
\vspace{5pt}  
\section{Experimental Evaluation}
\begin{table}[t]
\centering
\caption{Statistics of datasets.}
\label{tab:datasets}    
\renewcommand{\arraystretch}{1.15}
\resizebox{1\linewidth}{!}{%
\begin{tabular}{lcccccc}
\toprule
\textbf{Dataset} & \textbf{\#Nodes} & \textbf{\#Node Types} & \textbf{\#Edges} & \textbf{\#Edge Types} & \textbf{Target} & \textbf{\#Classes} \\ 
\midrule
American  & 9,473  & 7  & 465,557  & 13  & person  & 2 \\ 
Liar  & 14,395 & 4 & 45,358 & 6  & news & 6 \\ 
Actor & 16,255 & 3 & 76,118 & 18 & star & 7 \\ 
Amherst  & 3,422  & 7  & 190,277  & 13  & person  & 2 \\ 
MIT  & 9,274  & 7  & 532,102  & 13  & person  & 2 \\  
\bottomrule
\end{tabular}}
\vspace{-1em} 
\end{table}

\subsection{Experimental Setup}
\noindent\textbf{Datasets.} 
We evaluated our method on five real-world HG datasets, as detailed in Table~\ref{tab:datasets}:
\texttt{Liar\footnote{\url{https://huggingface.co/datasets/liar}}}~\cite{wang2017liar} is a dataset for Fake News Detection, aiming to predict news node categories.
\texttt{American, Amherst, and MIT} are three datasets selected from the FB100\footnote{\url{https://archive.org/details/oxford-2005-facebook-matrix}}  collection~\cite{traud2012social}, representing heterogeneous networks of Facebook ``friendship'' structures across 100 American colleges and universities, captured at a single point in time.
\texttt{Actor\footnote{\url{https://lfs.aminer.cn/lab-datasets/soinf/}}}~\cite{tang2009social} is a film-related network, focusing on predicting actor node categories.

\noindent\textbf{Baseline Methods.} 
We evaluate our method against 13 baselines, including 9 graph neural networks and 4 heterogeneous graph structure learning approaches. The 9 GNNs comprise 4 homogeneous and 5 heterogeneous models. \textbf{Homogeneous GNNs} include: GCN~\cite{kipf2016semi}, a seminal graph convolutional network; GAT~\cite{velivckovic2017graph}, leveraging attention for neighbor aggregation; H2GCN~\cite{zhu2020beyond}, addressing heterophily via higher-order neighborhoods; and LINKX~\cite{lim2021large}, which decouples structural and feature transformations.\textbf{ Heterogeneous GNNs} include: RGCN~\cite{schlichtkrull2018modeling}, assigning relation-specific parameters; HAN~\cite{wang2019heterogeneous}, utilizing hierarchical attention along meta-paths; SHGN~\cite{lv2021we}, applying relation-aware attention with normalization; HINormer~\cite{mao2023hinormer}, combining local and relation encoders; and PSHGCN~\cite{he2024spectral}, employing spectral graph filters for enhanced classification. \textbf{Graph Structure learning and Rewiring methods} consist of LDS~\cite{franceschi2019learning} and IDGL~\cite{chen2020iterative}, which iteratively optimize graph topology via task-driven bilevel optimization; HGSL~\cite{zhao2021heterogeneous}, integrating attribute and meta-path information for relation learning; and HDHGR~\cite{guo2023homophily}, a homophily-oriented rewiring strategy enhancing HGNN performance.

\noindent\textbf{Implementation Setup.} Datasets are split into training, validation, and testing sets with a 0.6/0.2/0.2 ratio, and models are trained for 400 epochs. The threshold \(\tau\) is searched within [0, 1] at intervals of 0.05. For Phase II, the number of fine-tuning epochs is set to 200, with \( \alpha \) in \([5\%, 10\%, 15\%]\), and gap factor \( \gamma = 0.1\). The base offset \( \mathcal{O} \) is randomly set from [0.03, 0.06, 0.09]. Pruning thresholds are iteratively set with intervals of \( \gamma \), and the process runs for 3 iterations. All experiments are conducted on two 32GB V100 GPUs. For the baseline method, we follow the settings from HDHGR~\cite{guo2023homophily}. PSHGCN~\cite{he2024spectral} follows the original settings. The Adam optimizer with a learning rate of \(5e-4\) and weight decay of \(1e-4\) is used.

\noindent\textbf{Evaluation Metrics.} We performed node classification tasks on HGNNs, using Macro-F1 and Micro-F1 scores as evaluation metrics. We calculated the average relative improvement (ARI) for all models on each dataset as an additional metric.
\begin{equation}\label{eqn-ARI}
\mathrm{ARI}=\frac{1}{|\mathcal{M}|}\sum_{m_i\in\mathcal{M}}\left(\frac{\mathrm{ACC}(m_i(\hat{\mathcal{G}}))-\mathrm{ACC}(m_i(\mathcal{G}))}{\mathrm{ACC}(m_i(\mathcal{G}))}\right).
\end{equation}
\vspace{-12pt}  

\begin{table*}[ht]\large
\renewcommand{\arraystretch}{1.25405}
\centering
\caption{Node classification results on HG datasets.The bold numbers indicate that our method improves the HGNN model.}
\label{tab:main_result}    
\resizebox{0.99\linewidth}{!}{
\begin{tabular}{cccccccccccc}
\toprule
\multirow{2}{*}{{Method}} & \multirow{2}{*}{{/}} & \multicolumn{2}{c}{{Liar}} & \multicolumn{2}{c}{{Actor}} & \multicolumn{2}{c}{{American}} & \multicolumn{2}{c}{{Amherst}} & \multicolumn{2}{c}{{MIT}} \\
\cline{3-12}
 &  & {Macro-F1} & {Micro-F1} & {Macro-F1} & {Micro-F1} & {Macro-F1} & {Micro-F1} & {Macro-F1} & {Micro-F1} & {Macro-F1} & {Micro-F1} \\
\midrule
\multirow{2}{*}{GCN} & origin & 20.75 ± 1.40 & 23.40 ± 1.97 & 48.78 ± 1.30 & 61.85 ± 1.41 & 57.30 ± 7.24 & 66.86 ± 3.35 & 71.33 ± 1.77 & 71.34 ± 1.78 & 65.40 ± 8.02 & 70.93 ± 3.76 \\ 
 & \textbf{CTHGE} & \textbf{{24.84 ± 1.29}} & \textbf{{25.09 ± 0.42}} & \textbf{{53.02 ± 1.01}} & \textbf{{63.57 ± 0.37}} & \textbf{{69.30 ± 0.84}} & \textbf{{70.62 ± 0.71}} & \textbf{{74.38 ± 1.19}} & \textbf{{74.08 ± 1.18}} & \textbf{{75.67 ± 1.15}} & \textbf{{77.71 ± 0.55}} \\
 \hline
\multirow{2}{*}{GAT} &origin	& 23.81 ± 1.66	& 24.43 ± 1.98	& 45.65 ± 1.03	& 61.26 ± 1.02	& 64.19 ± 0.30	& 68.05 ± 0.37	& 71.73 ± 1.12	& 71.74 ± 1.26	& 73.99 ± 1.42	& 75.16 ± 1.26 \\ 
 &\textbf{CTHGE}	& \textbf{24.52 ± 0.22}	& \textbf{24.74 ± 0.25}	& \textbf{50.81 ± 0.65}	& \textbf{62.01 ± 0.92} 	& \textbf{73.66 ± 1.16}	& \textbf{76.39 ± 1.08}	& \textbf{74.27 ± 0.36}	& \textbf{74.28 ± 0.37}	& \textbf{77.50 ± 1.51}	& \textbf{78.60 ± 1.31} \\ 
 \hline
\multirow{2}{*}{H2GCN} &origin	& 19.84 ± 0.31	& 22.00 ± 1.60	& 51.42 ± 1.21	& 63.13 ± 1.34	& 75.73 ± 0.29	& 77.33 ± 0.49	& 80.74 ± 0.37	& 80.75 ± 0.31	& 78.69 ± 2.13	& 79.48 ± 1.91 \\ 
 &\textbf{CTHGE}	& \textbf{22.03 ± 1.19}	& \textbf{23.19 ± 1.36}	& \textbf{53.20 ± 1.39}	&\textbf{63.49 ± 1.62}	& \textbf{76.52 ± 0.21}	& \textbf{78.33 ± 0.56}	& \textbf{81.89 ± 1.85}	& \textbf{81.90 ± 1.80}	& \textbf{81.55 ± 1.16}	& \textbf{82.55 ± 1.10} \\ 
 \hline
\multirow{2}{*}{LINKX} &origin	& 17.26 ± 1.43	& 19.04 ± 2.74	& 50.60 ± 1.02	& 55.83 ± 1.12	& 74.47 ± 0.45	& 74.81 ± 0.49	& 81.88 ± 1.99	& 81.88 ± 1.98	& 77.06 ± 2.07	& 77.28 ± 2.20 \\ 
&\textbf{CTHGE}	& \textbf{19.47 ± 2.12}	& \textbf{20.13 ± 2.05}	&\textbf{52.81 ± 1.37}	&\textbf{57.77 ± 0.98} 	& \textbf{77.60 ± 0.56}	& \textbf{79.27 ± 0.75}	& \textbf{83.12 ± 1.64}	& \textbf{83.13 ± 1.63}	& \textbf{80.40 ± 0.39}	& \textbf{80.91 ± 0.40} \\ 
\hline
\multirow{2}{*}{RGCN} &origin	& 17.71 ± 0.53	& 22.58 ± 2.10	& 65.62 ± 1.34	& 78.31 ± 1.01	& 52.57 ± 2.77	& 64.03 ± 1.12	& 60.74 ± 3.43	& 60.93 ± 3.22	& 67.17 ± 3.12	& 69.24 ± 2.85 \\ 
&\textbf{CTHGE}	& \textbf{20.34 ± 1.13}	& \textbf{22.95 ± 0.95}	& \textbf{71.36 ± 1.41}	&\textbf{82.48 ± 0.92}	&\textbf{ 55.22 ± 1.90}	& \textbf{63.65 ± 0.52}	& \textbf{62.72 ± 2.82}	& \textbf{62.90 ± 2.65}	& \textbf{69.08 ± 2.05}	&\textbf{ 70.79 ± 2.20} \\ 
\hline
\multirow{2}{*}{HAN} &origin & 26.08 ± 0.01 & 22.47 ± 0.74 & 69.43 ± 0.22 & 56.54 ± 0.94 & 63.95 ± 0.08 & 56.22 ± 0.43 & 61.34 ± 0.28 & 61.16 ± 0.41 & 71.42 ± 0.97 & 69.93 ± 0.85 \\ 
 &\textbf{CTHGE} & \textbf{26.92 ± 0.21} & \textbf{26.52 ± 0.34} & \textbf{82.58 ± 0.39} & \textbf{73.62 ± 0.61} & \textbf{64.59 ± 0.87} & \textbf{57.81 ± 0.79} & \textbf{69.73 ± 5.79} & \textbf{69.69 ± 5.85} & \textbf{72.72 ± 0.99} & \textbf{71.36 ± 1.17} \\
\hline
\multirow{2}{*}{SHGN} &origin	& 21.31 ± 1.44	& 24.35 ± 1.82	& 70.63 ± 0.93	& 78.95 ± 1.18	& 75.17 ± 3.65	& 77.22 ± 2.83	& 79.77 ± 1.69	& 79.85 ± 1.72	& 77.91 ± 2.22	& 79.05 ± 1.99 \\ 
 &\textbf{CTHGE}	& \textbf{24.42 ± 0.06}	& \textbf{25.66 ± 0.51}	& \textbf{73.11 ± 0.85}	& \textbf{81.29 ± 0.74}	& \textbf{78.15 ± 0.59}	& \textbf{79.55 ± 0.49}	& \textbf{82.88 ± 3.56}	& \textbf{82.81 ± 3.50}	& \textbf{81.12 ± 0.41}	& \textbf{82.06 ± 0.56} \\ 
 \hline
\multirow{2}{*}{HINormer} &origin     & 20.29 ± 1.13     & 	23.18 ± 1.70     & 	47.72 ± 1.34	&63.05 ± 0.59	&66.57 ± 1.63	& 70.02 ± 1.54	& 72.09 ± 1.44	& 72.15 ± 1.50	& 73.15 ± 2.26	& 74.48 ± 2.46 \\ 
 & \textbf{CTHGE}	& \textbf{24.60 ± 0.39}	& \textbf{25.10 ± 0.35}	& \textbf{52.87 ± 1.76}	& \textbf{65.71 ± 0.62}	& \textbf{75.62 ± 1.10}	& \textbf{77.53 ± 1.04}	& \textbf{80.90 ± 2.43}	& \textbf{80.92 ± 2.42}	& \textbf{77.26 ± 1.28}	& \textbf{78.54 ± 1.05} \\ 
  \hline
 \multirow{2}{*}{PSHGCN} &origin & 20.42 ± 0.25 & 15.13 ± 0.23 & 76.20 ± 1.28 & 67.16 ± 1.12 & 64.25 ± 0.25 & 56.82 ± 1.48 & 63.88 ± 0.25 & 63.88 ± 0.46 & 65.32 ± 0.98 & 60.77 ± 0.56 \\ 
 &\textbf{CTHGE} & \textbf{42.52 ± 3.12} & \textbf{42.08 ± 3.09} & \textbf{77.58 ± 1.85} & \textbf{68.59 ± 1.92} & \textbf{65.33 ± 0.87} & \textbf{61.87 ± 1.14} & \textbf{65.15 ± 1.51} & \textbf{65.14 ± 1.43} & \textbf{67.64 ± 0.76} & \textbf{64.05 ± 0.52} \\
 \hline
 ARI & {/} & \textbf{23.19\%$\uparrow$} & \textbf{25.67\%$\uparrow$} & \textbf{7.96\%$\uparrow$} & \textbf{5.88\%$\uparrow$} & \textbf{7.36\%$\uparrow$} & \textbf{5.29\%$\uparrow$} & \textbf{5.09\%$\uparrow$} & \textbf{5.04\%$\uparrow$} & \textbf{5.15\%$\uparrow$} & \textbf{4.63\%$\uparrow$} \\

\bottomrule
\end{tabular}
}
\end{table*}

\begin{table}[!t]
\centering
\caption{Performance comparison of graph structure learning methods.}
\label{tab:GLS}    
\resizebox{0.99005\linewidth}{!}{%
\begin{tabular}{lcccc}
\toprule
\multirow{2}{*}{{Method}} & \multicolumn{2}{c}{{Liar}} & \multicolumn{2}{c}{{American}}  \\
\cline{2-5}
 & {Macro-F1} & {Micro-F1} & {Macro-F1} & {Micro-F1}  \\
\midrule
Base Model & 20.75 ± 1.40 & 23.40 ± 1.97 & 57.30 ± 7.24 & 66.86 ± 3.35  \\
LDS & 20.95 ± 1.87 & 22.92 ± 1.70 & 58.42 ± 5.98 & 66.92 ± 3.44  \\
IDGL & 21.99 ± 1.71 & 23.02 ± 1.92 & 58.85 ± 4.34 & 67.10 ± 2.20 \\
HGSL & 21.01 ± 0.89 & 23.10 ± 1.41 & 58.93 ± 2.11 & 67.01 ± 1.87 \\
HDHGR & 23.01 ± 0.63 & 24.19 ± 0.81 & 65.72 ± 2.06 & 67.92 ± 0.94 \\
\textbf{{CTHGE}} & \textbf{24.84 ± 1.29} & \textbf{25.09 ± 0.42} & \textbf{69.30 ± 0.84} & \textbf{70.62 ± 0.71}  \\
\bottomrule
\end{tabular}}
\end{table}

\subsection{Main Results}
We conducted node classification experiments on five HG datasets using nine representative HGNN models to evaluate the effectiveness of our proposed method. As detailed in Table~\ref{tab:main_result}, the CTHGE approach consistently achieves significant performance improvements across all baseline HGNN backbones, with relative gains ranging from 4.63\% to 25.67\%. The CTHGE method is structurally simple and involves less hyperparameter tuning, facilitating its integration with diverse HGNN architectures. The performance improvement observed on HAN and PSHGCN methods demonstrates that our approach is also suitable for metapath-based HGNN methods. In the PSHGCN method, long metapath information is utilized, and our method also brings performance improvement to PSHGCN, indicating that our approach does not negatively impact the performance of methods that rely on long metapaths.

As shown in Fig.~\ref{fig:CHR_comparison}, we analyze changes in Cross-Type Homophily Ratio (CHR) before and after our editing. The results demonstrate that CTHGE effectively increases the CHR, enhancing semantic alignment between cross-type nodes. This improvement promotes more efficient information propagation and feature aggregation, thereby boosting overall HGNN performance. Notably, CTHGE attains the highest average relative improvement (ARI) on the Liar dataset, which exhibits a low initial CHR of only 21.66\%, indicating pronounced heterophily among cross-type edges. Theoretical analysis suggests that structural optimization is particularly crucial in low-homophily settings, as it suppresses noisy or irrelevant connections while reinforcing task-relevant paths. This explains the superior performance of CTHGE on such graphs and highlights its robustness and applicability across heterogeneous structures.

\subsection{Comparison with Graph Structure Learning Methods}
We compare CTHGE with graph structure learning methods, using multi-GCN, an extension of GCN that maps node features from various categories into a unified feature space. As shown in Table~\ref{tab:GLS}, our method consistently outperforms the others, demonstrating its effectiveness. Our approach surpasses the state-of-the-art baselines in supervised structural modification for HGs. Our approach significantly enhances HGNNs performance without increasing graph complexity, thereby facilitating its efficient use in practical applications. This result emphasizes the effectiveness of our approach and highlights the potential of focusing on cross-type edges in HGs as a compelling avenue for future research.

\begin{figure}[!t]
	\centering
    \setlength{\abovecaptionskip}{-0.03cm}
	\setlength{\belowcaptionskip}{0cm}
	\includegraphics[width=0.86\linewidth]{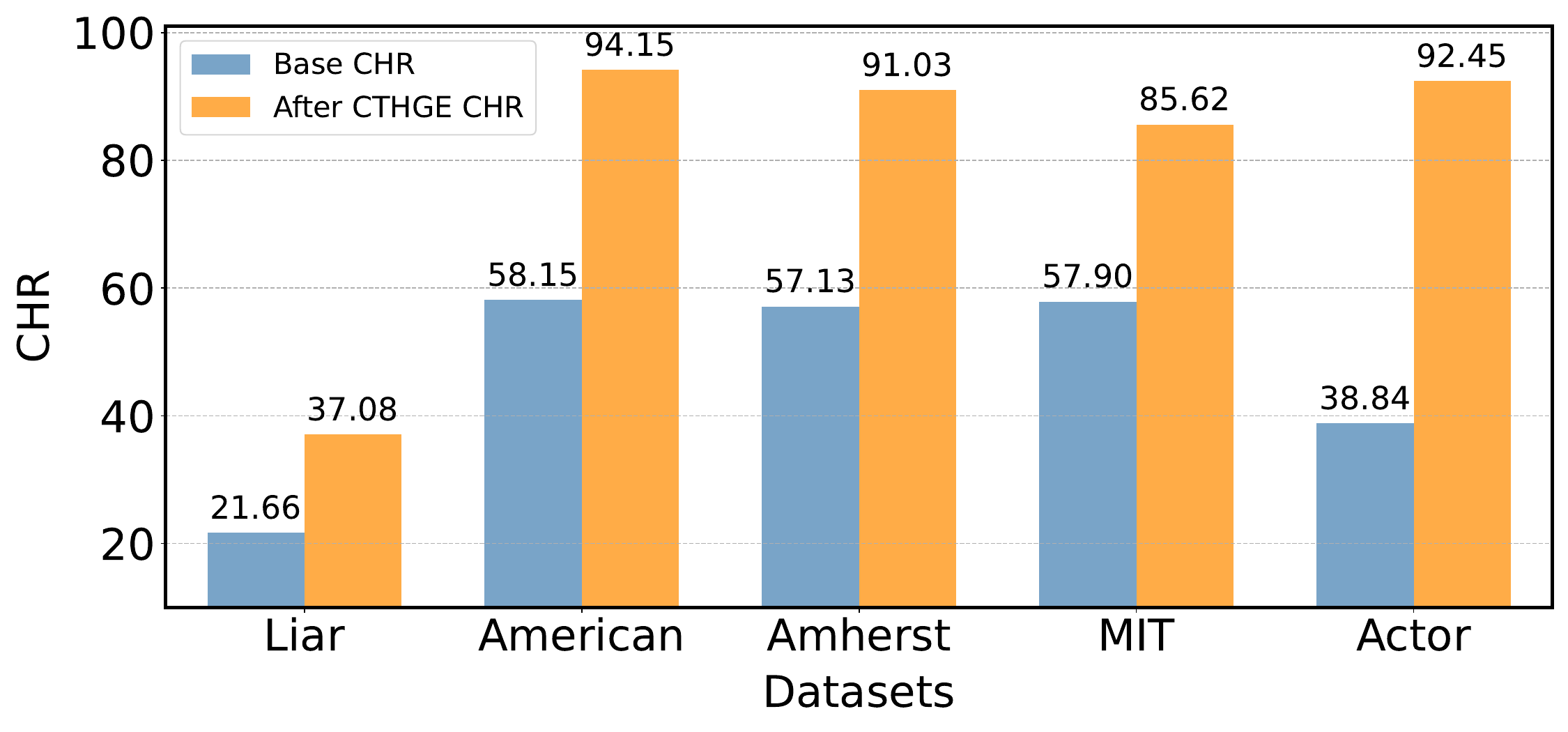}
    \caption{CHR comparison before and after CTHGE. \textit{Base CHR} represents the CHR value of the HG before CTHGE, while \textit{After CTHGE CHR} denotes the CHR value of the HG after applying CTHGE method.}
	\label{fig:CHR_comparison}
\end{figure}

\begin{table*}[t]\large
\renewcommand{\arraystretch}{1.08}
\centering
\caption{Ablation study analysis of CTHGE.}
\label{tab:AbofCTHGE}    
\resizebox{0.99\linewidth}{!}{%
\begin{tabular}{lcccccccccc}
\toprule    
\multirow{2}{*}{{Method}} & \multicolumn{2}{c}{{Liar}}  & \multicolumn{2}{c}{{Actor}} & \multicolumn{2}{c}{{American}}  & \multicolumn{2}{c}{{Amherst}}  & \multicolumn{2}{c}{{MIT}}\\
\cline{2-11}
 & {Macro-F1}    & {Micro-F1}   & {Macro-F1}    & {Micro-F1}   & {Macro-F1}    & {Micro-F1} & {Macro-F1}    & {Micro-F1}    & {Macro-F1}    & {Micro-F1}\\
\midrule
{Base Model} & 20.75 ± 1.40 & 23.40 ± 1.97 & 48.78 ± 1.30 & 61.85 ± 1.41 & 57.30 ± 7.24 & 66.86 ± 3.35 & 71.33 ± 1.77 & 71.34 ± 1.78 & 65.40 ± 8.02 & 70.93 ± 3.76 \\
{RandDropEdge} & 19.30 ± 1.48 & 22.68 ± 1.30 & 48.34 ± 1.32 & 61.93 ± 1.48 & 56.12 ± 7.39 & 66.03 ± 3.63 & 70.54 ± 1.47 & 70.51 ± 1.48 & 64.53 ± 7.36 & 69.85 ± 3.31 \\
{{CTHGE(w/o Phase II)}} & {24.01 ± 1.76} & {24.82 ± 0.82} & {52.53 ± 1.22} & {63.02 ± 1.13} & {68.77 ± 0.99} & {69.73 ± 1.01} & {73.78 ± 0.38} & {73.69 ± 0.38} & {75.04 ± 0.24} & {77.26 ± 0.30} \\
{{CTHGE(w/o PT)}}       & {24.13 ± 1.83} & {24.96 ± 0.71} & {52.70 ± 1.18} & {63.24 ± 1.26} & {68.81 ± 1.01} & {69.80 ± 0.98} & {73.98 ± 0.47} & {73.77 ± 0.34} & {75.49 ± 1.92} & {77.43 ± 0.22} \\
{{CTHGE(w/o IL)}}       & {24.52 ± 1.49} & {25.00 ± 0.52} & {52.87 ± 1.30} & {63.30 ± 1.03} & {68.99 ± 0.78} & {69.89 ± 1.32} & {74.12 ± 0.62} & {73.80 ± 0.45} & {75.53 ± 0.76} & {77.58 ± 0.41} \\
{\textbf{{CTHGE}}} & \textbf{{24.84 ± 1.29}} & \textbf{{25.09 ± 0.42}} & \textbf{{53.02 ± 1.01}} & \textbf{{63.57 ± 0.37}} & \textbf{{69.30 ± 0.84}} & \textbf{{70.62 ± 0.71}} & \textbf{{74.38 ± 1.19}} & \textbf{{74.08 ± 1.18}} & \textbf{{75.67 ± 1.15}} & \textbf{{77.71 ± 0.55}} \\
\bottomrule
\end{tabular}}
\end{table*}

\subsection{Ablation Study}
We conduct an ablation study to evaluate the effectiveness of our editing framework. Using multi-GCN as the base model, we introduce a random graph structure transformation technique, RandDropEdge, which randomly removes edges from the original graph with a probability of 50\%, to assess its impact on model performance. As shown in Table~\ref{tab:AbofCTHGE}, we observe that random pruning results in a performance decline across all models, whereas our CTHGE framework yields significant improvements. CTHGE enhances graph editing confidence through its two-phase approach, demonstrating greater improvements.

To further investigate the contributions of each component, we explore the effect of two phase by conducting additional settings. ``w/o Phase II'', which removes the second phase and only uses the first phase, showing that the direct pruning in the first phase alone leads to substantial improvements. Additionally, we examine the impact of two specific strategies. ``w/o PT'', which removes the Target-Driven Auxiliary Learning Paradigm Task, and ``w/o IL'', which removes Iterative Logits Refinement. The results indicate that the Paradigm Task plays a key role in driving performance improvement, highlighting that our proposed Paradigm Task, which leverages target information, is highly effective. By incorporating the pre-training and fine-tuning paradigm, we effectively utilize target information in node classification tasks to achieve a more confident and optimal edge set.

\begin{figure}[!t]
    \centering
	\begin{subfigure}{0.495\linewidth}
		\centering
		\includegraphics[width=1\linewidth]{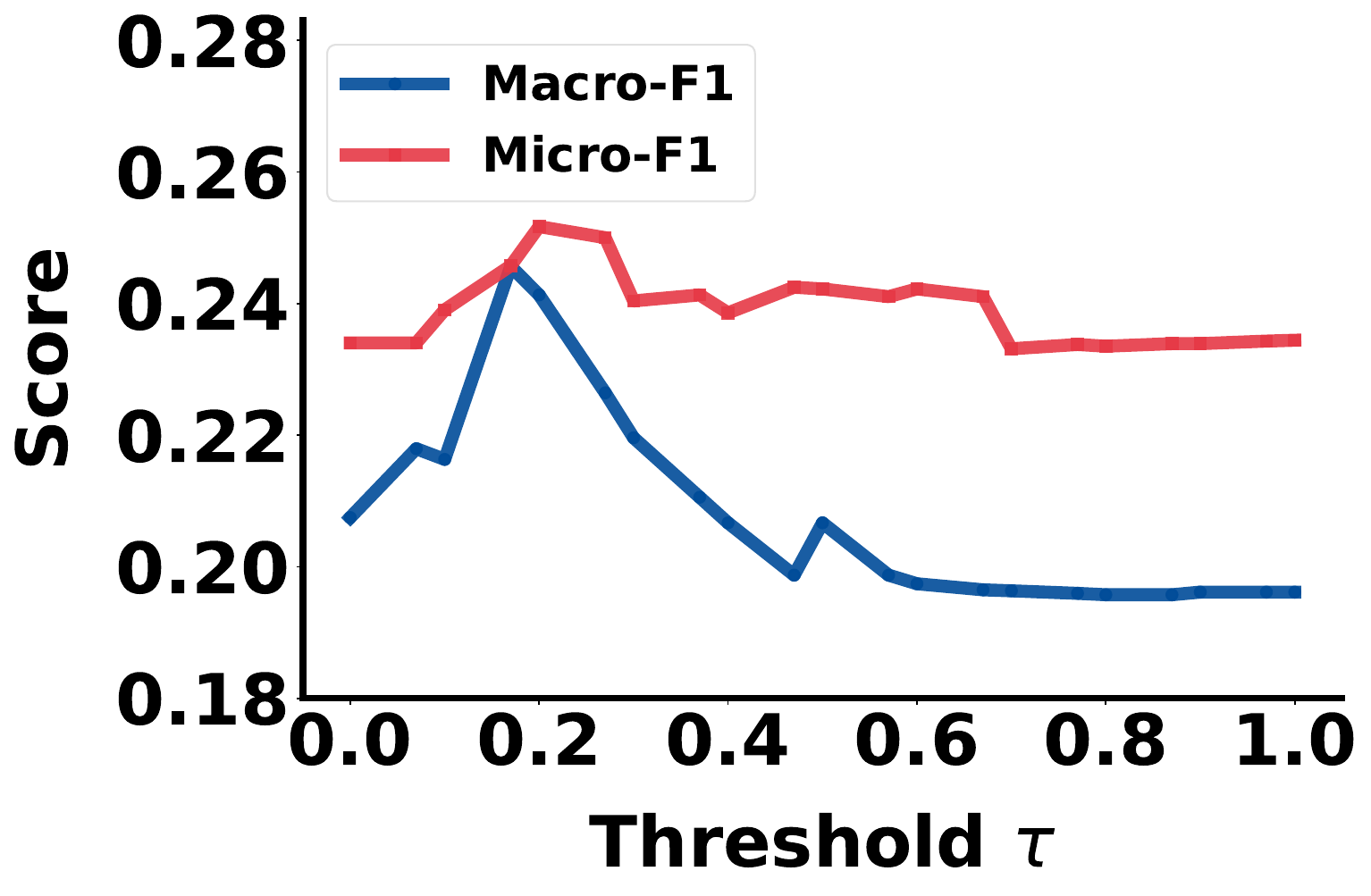}
		\caption{Liar}
		\label{fig-exp_Threshold_F1_Liar}
	\end{subfigure}
	\hfill
	\begin{subfigure}{0.495\linewidth}
		\centering
		\includegraphics[width=1\linewidth]{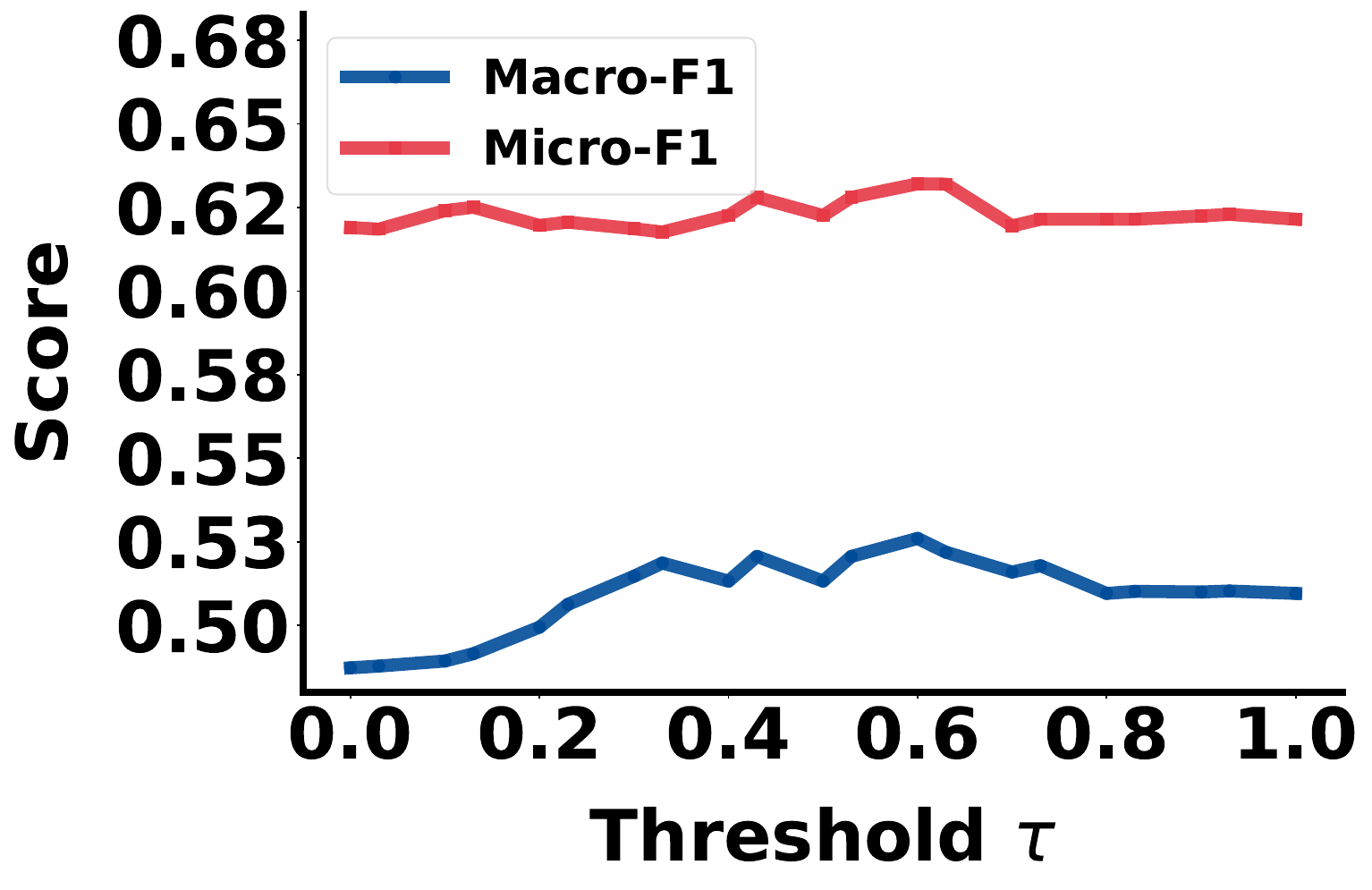}
		\caption{Actor}
		\label{fig-exp_Threshold_F1_Actor}
	\end{subfigure}
    \vskip\floatsep
    \centering
	\begin{subfigure}{0.495\linewidth}
		\centering
		\includegraphics[width=1\linewidth]{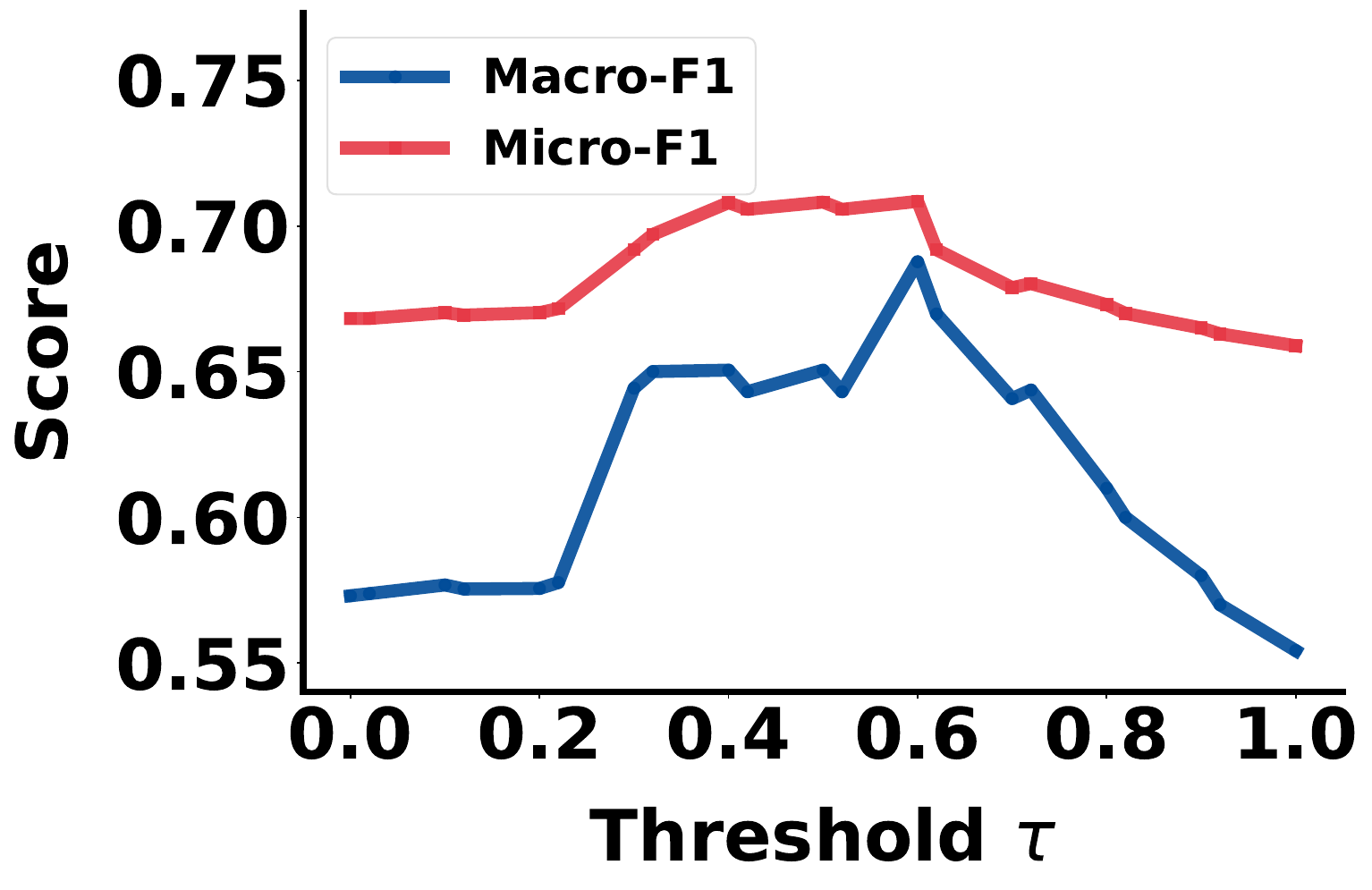}
		\caption{American}
		\label{fig-exp_Threshold_F1_American}
	\end{subfigure}
    \hfill
	\begin{subfigure}{0.495\linewidth}
		\centering
		\includegraphics[width=1\linewidth]{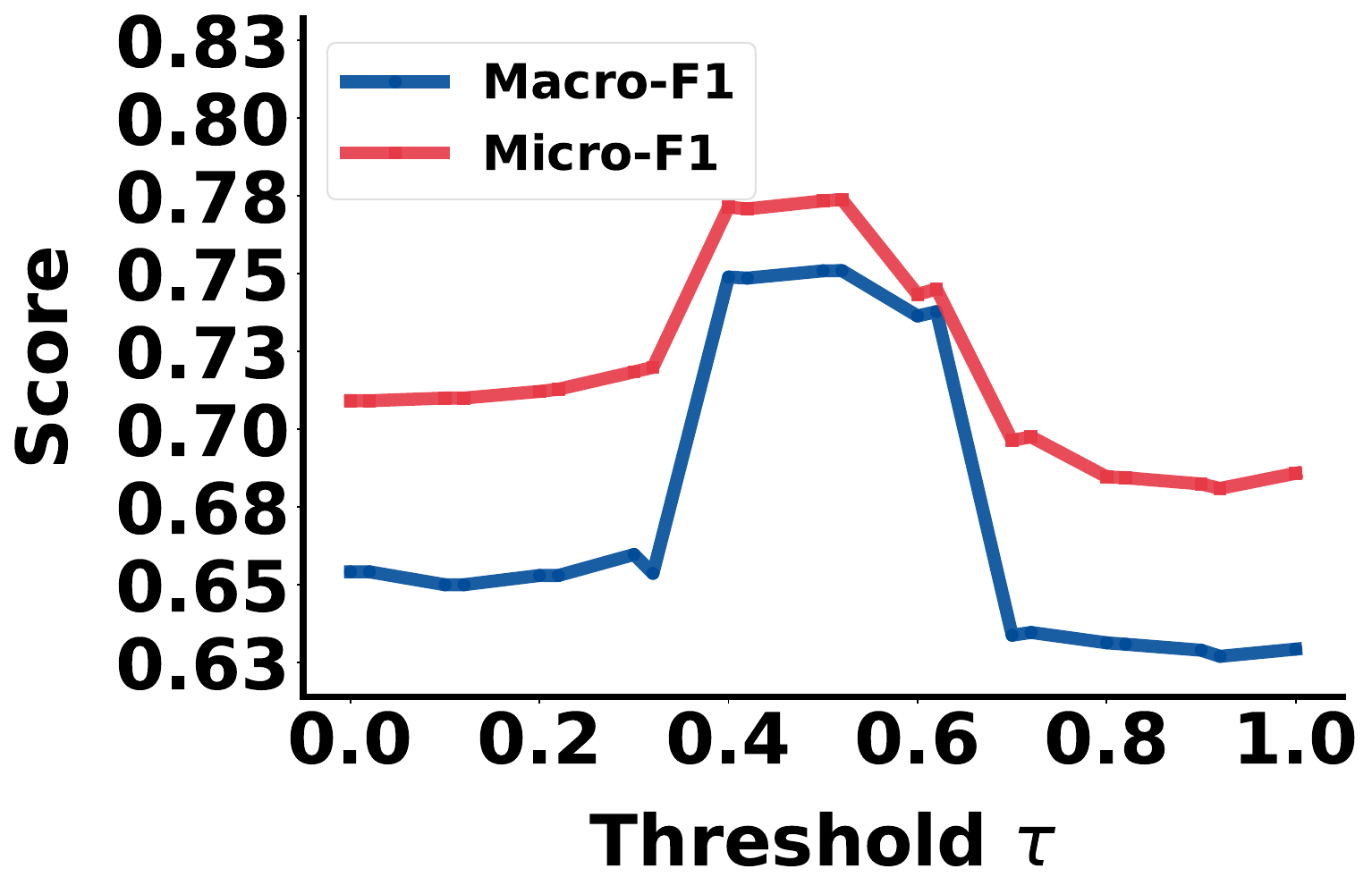}
		\caption{MIT}
		\label{fig-exp_Threshold_F1_MIT}
	\end{subfigure}
	\caption{Hyper-parameter study of $\tau$ and F1.}
	\label{fig:exp_Threshold_F1}
\end{figure}

\begin{figure}[!t]
    \centering
	\begin{subfigure}{0.49\linewidth}
		\centering
		\includegraphics[width=1\linewidth]{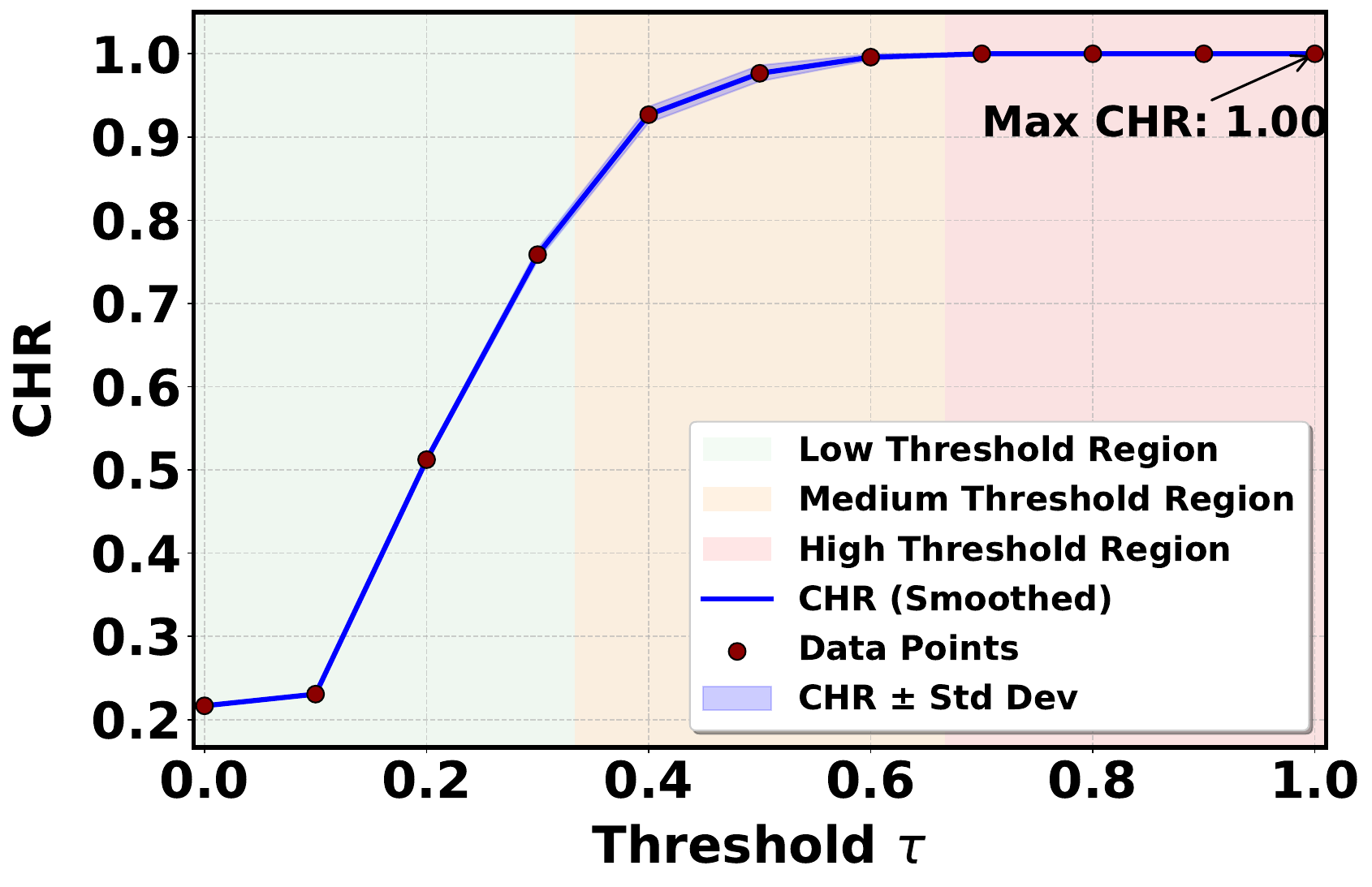}
		\caption{Liar}
		\label{fig-exp_Threshold_CHR_Liar}
	\end{subfigure}
	\hfill
	\begin{subfigure}{0.49\linewidth}
		\centering
		\includegraphics[width=1\linewidth]{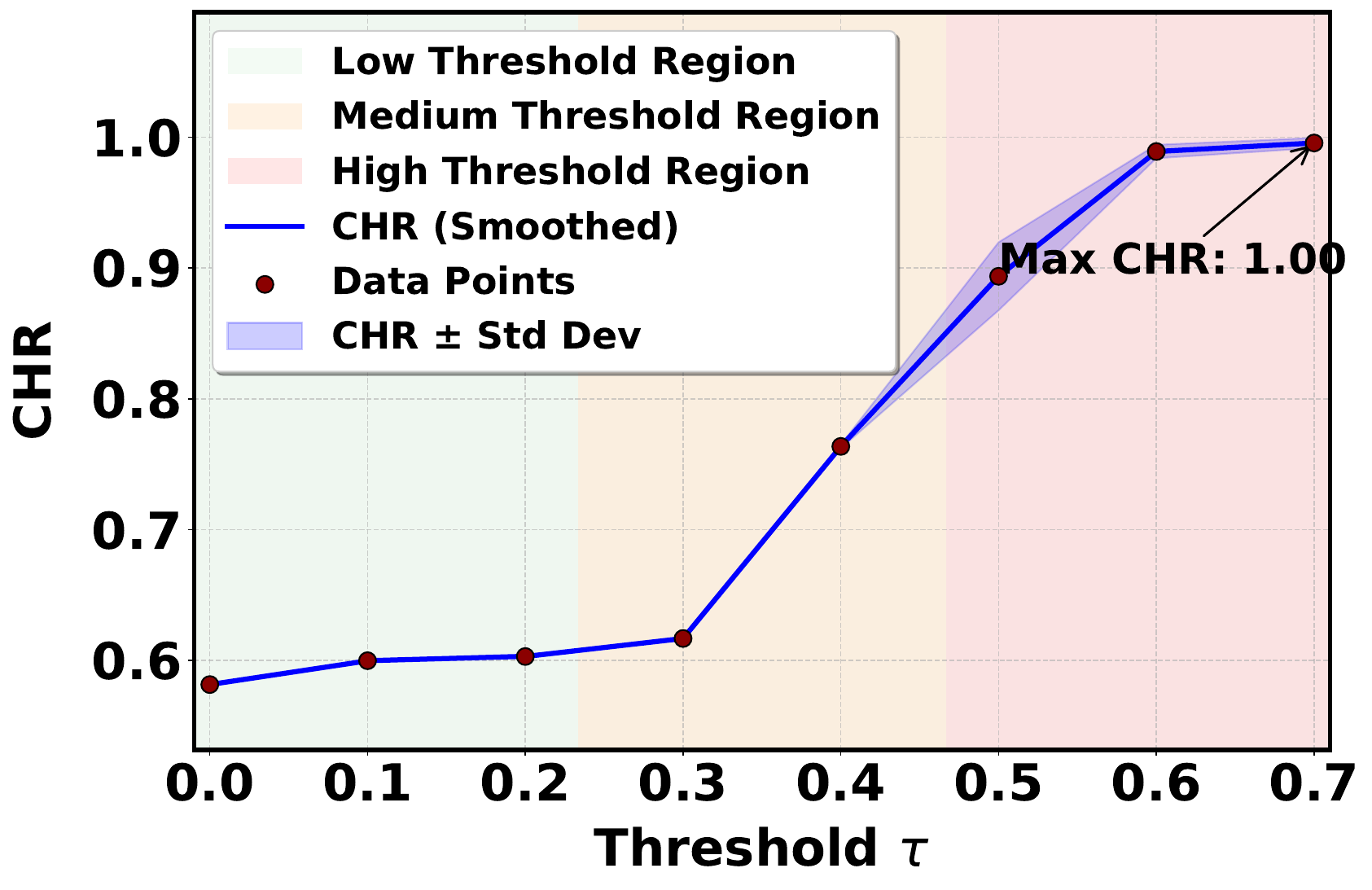}
		\caption{American}
		\label{fig-exp_Threshold_CHR_American}
	\end{subfigure}

    \vskip\floatsep

    \centering
	\begin{subfigure}{0.493\linewidth}
		\centering
		\includegraphics[width=1\linewidth]{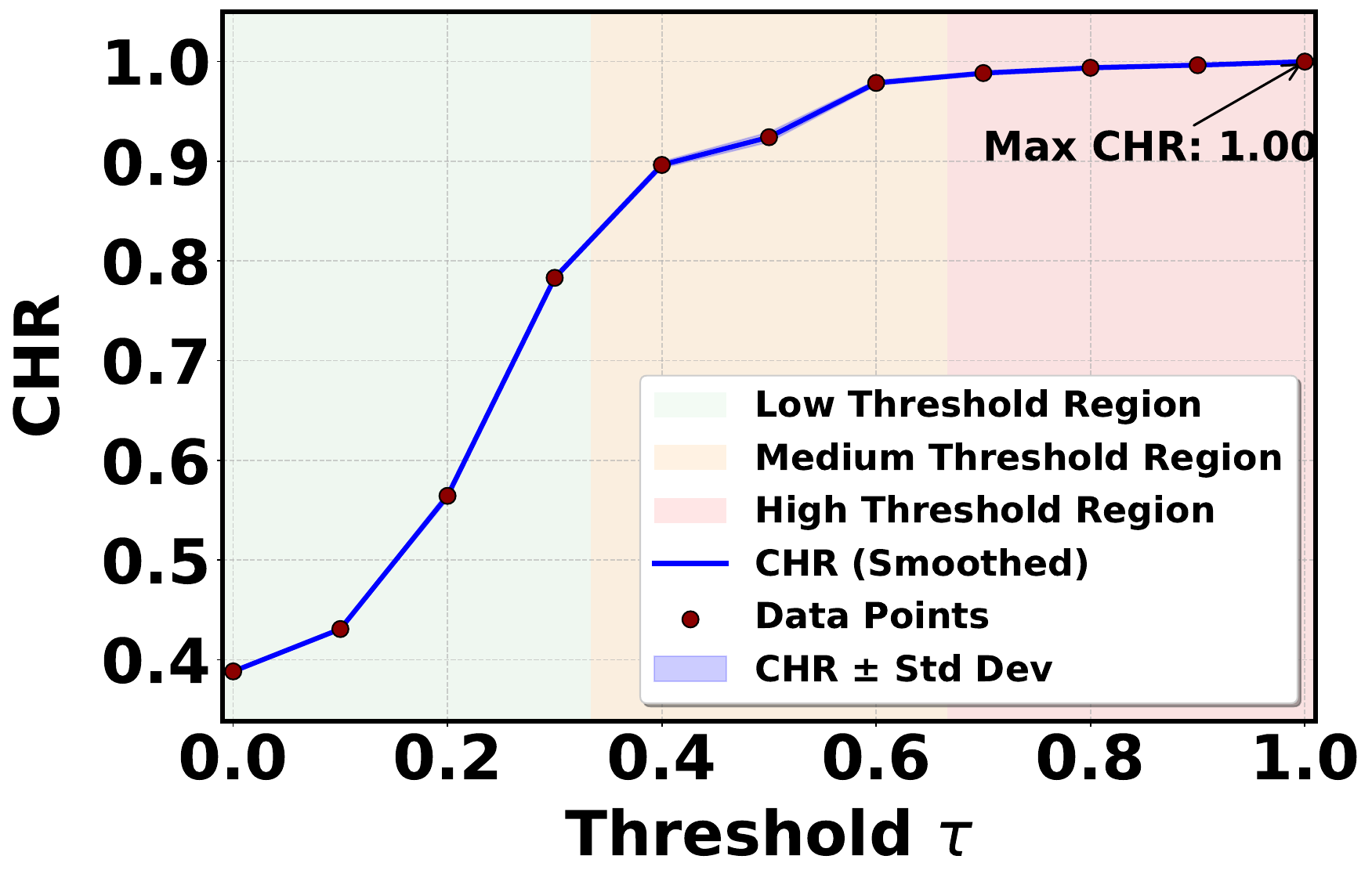}
		\caption{Actor}
		\label{fig-exp_Threshold_CHR_Actor}
	\end{subfigure}
    \hfill
	\begin{subfigure}{0.493\linewidth}
		\centering
		\includegraphics[width=1\linewidth]{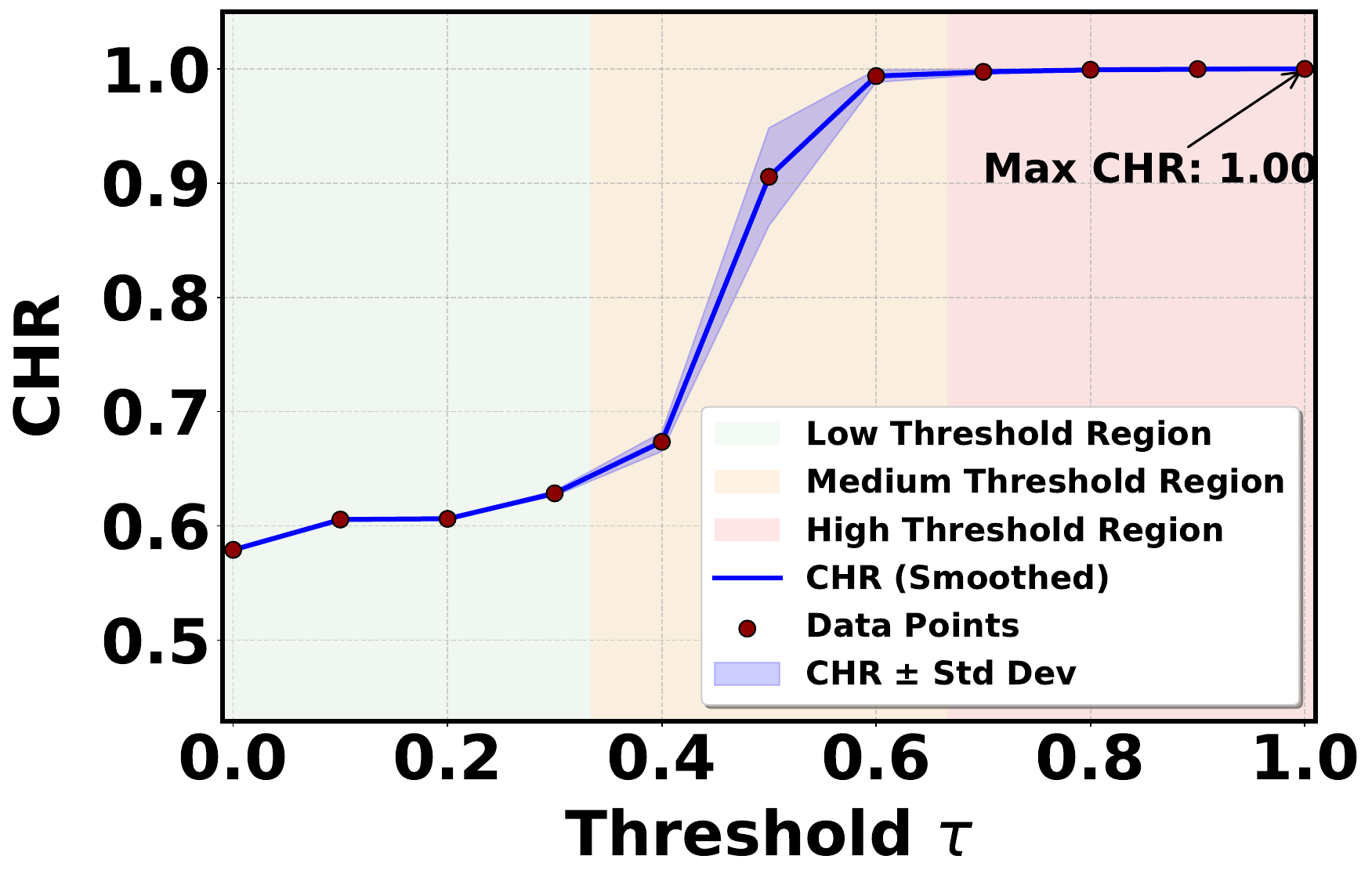}
		\caption{MIT}
		\label{fig-exp_Threshold_CHR_MIT}
	\end{subfigure}
	\caption{Hyper-parameter study of $\tau$ and CHR.}
	\label{fig:exp_Threshold_CHR}
\end{figure}

\subsection{Hyper-parameter Study}

Since the parameters for the second phase can generally follow a fixed strategy, we have already explored the effects of each component of the second phase in our ablation study. Therefore, in this hyper-parameter study, we primarily focus on examining the effect of the Threshold \(\tau\).

\noindent\textbf{Threshold \(\tau\) and Model Performance.} As shown in Fig.~\ref{fig:exp_Threshold_F1}, we observe an approximately unimodal relationship between the threshold \(\tau\) and model performance. The best performance is typically achieved when the threshold \(\tau\) falls within the range of 0.3 to 0.6. With a low threshold, heterophilic edges are not sufficiently pruned, which limits the effectiveness of the pruning algorithm. Conversely, a high threshold enforces strict edge retention, preserving only a small set of high-homophily edges, which can lead to a decrease in HGNN node classification performance.

\noindent\textbf{Threshold \(\tau\) and Cross-Type Homophily Ratio.} As shown in Fig.~\ref{fig:exp_Threshold_CHR}, we observe that as the threshold \(\tau\) increases, the CHR also increases, indicating a growing proportion of homophilous edges among the cross-type edges in the pruned dataset. However, model performance initially improves as the threshold \(\tau\) increases but subsequently declines, in contrast to CHR, which continues to increase. This divergence occurs because, beyond a certain threshold, the pruning process begins to eliminate many cross-type edges that exhibit homophily, retaining only a few with very high homophily. From the perspective of CHR, a low threshold \(\tau\) does not achieve sufficient pruning, while an excessively high threshold removes too many homophilous cross-type edges. A mid-range threshold \(\tau\) generally yields the optimal model performance.

\section{Conclusion}
This study introduces cross-type homophily ratio(CHR), a novel metric to quantify cross-type homophily in HGs, representing the first exploration of cross-type edges in this domain. We establish the theoretical foundations of CHR and validate its empirical effectiveness in HGNNs. Additionally, we propose CTHGE, a graph editing framework that enhances CHR and improves HGNN performance on HG node classification tasks. Serving as a versatile plug-in compatible with various HGNN architectures, CTHGE demonstrates its effectiveness in improving CHR and node classification performance through extensive experiments.

\section*{Acknowledgment}
This work is supported by the Natural Science Foundation of Guangdong Province, China (2024A1515110162).

\section*{GenAI Usage Disclosure}
Generative AI tools were used exclusively for minor language edits, such as grammar and clarity improvements, comparable to standard writing assistants. No content was generated or substantially rewritten using AI.

\bibliographystyle{ACM-Reference-Format}
\bibliography{sample-base}

\appendix

\end{document}